\documentclass[oneside,english]{amsart}

\usepackage{cmap}

\usepackage[a4paper,includeheadfoot,margin=3.64cm]{geometry}

\usepackage{beton}
\usepackage{euler}
\usepackage[T1]{fontenc}
\usepackage[latin9]{inputenc}
\usepackage{varioref}
\usepackage{amsthm}
\usepackage{amssymb}

\usepackage[usenames,dvipsnames]{xcolor}

\usepackage{hyperref}
\hypersetup{colorlinks=true,linkcolor=MidnightBlue,backref=page,citecolor=Magenta}

\usepackage[backend=bibtex,style=alphabetic,firstinits=true,doi=false]{biblatex}
\usepackage{mathtools}

\usepackage[shortlabels]{enumitem}

\makeatletter
\numberwithin{equation}{section} 
\numberwithin{figure}{section} 
\theoremstyle{plain}
\theoremstyle{plain}
\newtheorem{theorem}{Theorem}
  \theoremstyle{plain}
  \newtheorem{lemma}{Lemma}
  \theoremstyle{plain}
  \newtheorem{proposition}{Proposition}
  \theoremstyle{remark}
  \newtheorem*{note*}{Note}
  \theoremstyle{remark}
  \newtheorem*{conclusion*}{Conclusion}
  \theoremstyle{remark}
  \newtheorem{remark}{Remark}
 \theoremstyle{definition}
  \newtheorem{example}{Example}
  \newtheorem{definition}{Definition}
  \theoremstyle{plain}
  \newtheorem{corollary}{Corollary}

\usepackage{amsthm}
\usepackage{amsfonts}
\usepackage{amscd}
\usepackage[all]{xy}
\usepackage{pb-diagram,pb-xy}

\usepackage{tikz}
\usetikzlibrary{matrix,arrows}
\usetikzlibrary{shapes.geometric,fit}

\tikzstyle{legend_general}=[rectangle, rounded corners, thin,
                          top color= white,bottom color=lavander!25,
                          minimum width=2.5cm, minimum height=0.8cm,
                          violet]

\tikzset{
   dashellipse/.style={rectangle,draw,dashed,inner sep=8pt,fit={#1}}
}
                          
\usepackage{tikz-cd}
\usetikzlibrary{decorations.pathreplacing, positioning}

\usepackage{soul}

\newcommand{\g}{\mathfrak{g}}

\newcommand{\kf}{\mathfrak{k}}
\newcommand{\hf}{\mathfrak{h}}
\newcommand{\pf}{\mathfrak{p}}

\newcommand{\cI}{{\mathcal I}}

\newcommand{\cL}{{\mathcal L}}

\newcommand{\mR}{\mathbb{R}}

\newcommand{\gl}{\mathfrak{gl}}

\usepackage[e]{esvect}

\usepackage{charter}

\makeatother

\makeatletter
\newcommand*\bigcdot{\mathpalette\bigcdot@{.5}}
\newcommand*\bigcdot@[2]{\mathbin{\vcenter{\hbox{\scalebox{#2}{$\m@th#1\bullet$}}}}}
\makeatother

\usepackage{babel}

\begin{document}

\title[Chern-Simons, $3d$-gravity and extensions]{Chern-Simons field theory on the general affine group, $3d$-gravity and the extension of Cartan connections}

\author{S. Capriotti}
\address{Departamento de Matem\'atica \\
Universidad Nacional del Sur and Conicet\\
 8000 Bah{\'\i}a Blanca \\
Argentina}
\email{santiago.capriotti@uns.edu.ar}

\maketitle

\begin{abstract}
  The purpose of this article is to study the correspondence between $3d$-gravity and the Chern-Simons field theory from the perspective of geometric mechanics, specifically in the case where the structure group is the general affine group. To accomplish this, the paper discusses a variational problem of the Chern-Simons type on a principal fiber bundle with this group as its structure group. The connection to the usual Chern-Simons theory is established by utilizing a generalization, in the context of Cartan connections, of the notion of extension and reduction of connections.
\end{abstract}

\section{Integrable field theories and gravity}
\label{sec:chain-integr-field}


Chern-Simons field theory is a well-known type of gauge field theory \cite{10.2307/1971013,PhysRevLett.48.975,Freed:1992vw} whose quantization yields to topological field theory \cite{Witten1988,Ramadas1989}. In its more general setting \cite{Freed:1992vw}, the fields in Chern-Simons gauge theory are connections on any $K$-principal bundle with fixed base space. In this vein, the Chern-Simons action is considered as a function on the (infinite dimensional) manifold of the $K$-connections on a fixed base space $M$, and it is evident that the construction of this manifold requires a precise knowledge of every $K$-principal bundle on $M$. Although this operation can be performed successfully (getting more complicated as the dimension of $M$ increases), this scheme is out of range of the geometrical formulation for field theory \cite{Gotay:1997eg,Blee,doi:10.1142/9693,deLen2003,PROP:PROP2190440304}, because in this approach the fields should be sections of a definite bundle (preferably, of finite dimension). A successful formulation for Chern-Simons field theory fitting in this geometrical scheme is described in \cite{TejeroPrieto2004}, using a variational problem posed by local Lagrangians; a multisymplectic formulation for Chern-Simons field theory can be found in \cite{GOMIS2023116069}, where the structure of the constraints arising from the singular nature of the Lagrangian is studied. The formulation we will use in this article is inspired in \cite{wise2009symmetric}, and uses Cartan connections as fields.

\par On the other hand, it is an interesting result \cite{ACHUCARRO198689,MR974271} that when gauge group is the Poincar\'e group $\text{ISO}\left(2,1\right)$, Chern-Simon field theory in dimension $3$ can be related with Palatini gravity on a spacetime of the same dimension. This correspondence is achieved by the splitting
\[
  \mathfrak{iso}\left(2,1\right)=\mathfrak{so}\left(2,1\right)\oplus\mR^3,
\]
that decomposes the field $A$ in two parts, one living in $\mR^3$ and another in $\mathfrak{so}\left(2,1\right)$; the idea is to recognize each of these fields as a vierbein and a $\mathfrak{so}\left(2,1\right)$-connection respectively, which can be seen as the basic fields for the Palatini description of general relativity. In \cite{wise2009symmetric}, this scheme is generalized to pair of algebras $\g\subset\hf$ defining a Cartan connection on a $G$-principal bundle $\pi:P\to M$.

\par A remarkable fact is that, in some of these descriptions of Chern-Simons theory and its connection with gravity, is usual a certain lack of definiteness with respect to the nature of the principal bundle to which the connections of the theory belongs. While the corresponding base space and structure group are properly set, it is avoided any precision on the geometrical characteristics of such bundle. Not that it keeps people away from working with Chern-Simons gauge theory: It is in this context where frameworks like the one discussed by Freed are fruitful. Moreover, the informed reader could recall that when dealing with gravity, a metric is available, and it can be used to select suitable subbundles of the frame bundle. Nevertheless, this answer should be considered as partial, and in fact might put us in a paradoxical situation, because the metric is part of the dynamical fields in gravity, and so it should be necessary to solve the equations of motion of gravity before constructing the bundle where the fields should live. Therefore, it seems to appear a tension between the formalism describing field theory from a geometric viewpoint, and the characteristics that a Chern-Simons gauge theory must have in order to represent gravity. At the end, the apparent paradox is solved by invoking the gauge symmetry of the Lagrangian, but it could be interesting to explore how to deal with this situation from a geometrical point of view.


\par These considerations set the aims of the article. Basically, we are looking for a formulation of the correspondence between Chern-Simons field theory and gravity (for spacetime dimension $m=3$) where the structure groups of the involved principal bundle are respectively the affine general group $A\left(3\right)$   and the general linear group $GL\left(3\right)$. In order to achieve this objective, we will generalize the scheme posed by Wise. The main tool used for the generalization is the (as far as I know, novel) formulation of a concept equivalent to \emph{extension of a connection} for Cartan connections (see Section \ref{sec:extens-gener-cart} below). The idea for the generalization is that any Cartan connection can be seen as a principal connection on a suitable bundle (i.e. Proposition \ref{prop:CartanVsPrincipalConn}); thus the extension of a Cartan connection is the Cartan connection induced by the extension of this associated principal connection. As we will see, a drawback of this construction is that, in general, the connections obtained will be of more general nature than Cartan connections (are generalized Cartan connections, as defined in \cite{alekseevsky95:_differ_cartan}).

Let us describe in some detail how this scheme will be implemented. In the formulation of Wise, the basic geometrical data is a $K$-principal bundle $R_\zeta$, and its fields are described by a Cartan connection taking values in a Lie algebra $\g$ such that $\kf\subset\g$. Let us recall that the Lorentz group $K=SO\left(m-1,1\right)$ becomes a subgroup of different Lie groups $G$ in the Wise scheme, depending on the sign of the cosmological constant $\Lambda$; namely, we have that the total group $G$ is in each case
\[
  G=
  \begin{cases}
    SO\left(m,1\right)&\Lambda>0\\
    K\oplus\mR^m&\Lambda=0\\
    SO\left(m-1,2\right)&\Lambda<0.
  \end{cases}
\]
$K$ enters as subgroup in the case $\Lambda<0$ through the immersion
\[
  \iota_-\left(A\right)=
  \begin{bmatrix}
    A&0\\
    0&1
  \end{bmatrix}
\]
for every $A\in K$, and
\[
  \iota_+\left(A\right)=
  \begin{bmatrix}
    1&0\\
    0&A
  \end{bmatrix}
\]
in the case $\Lambda>0$. In the present article we will focus in the case $\Lambda=0$; the other cases will be considered elsewhere.

In order to see how to proceed for the generalization of this scheme to a principal bundle with structure group $GL\left(m\right)$, let us consider the following diagram of Lie groups:
\[
  \begin{tikzcd}[ampersand replacement=\&,row sep=1.5cm,column sep=1.5cm]
    {GL\left(m\right)}
    \arrow[hook]{d}{}
    \&
    {K}
    {\arrow[hook]{d}{}}
    \arrow[hook']{l}{}
    \\
    GL\left(m\right)\oplus\mR^m
    \&
    {K\oplus\mR^m}
    \arrow[hook']{l}{}
  \end{tikzcd}      
\]
Arrows are induced by canonical inclusions. We want to promote it to a commutative diagram involving principal bundles with structure groups borrowed from the nodes of this diagram. Now, what kind of $K$-principal bundle $R_\zeta$ on a manifold $M$ has a Cartan connection with values in the Lie algebra $\g=\kf\oplus\mR^m$? A possible answer can be found using geometrical considerations. In fact, as it is indicated in \cite{sharpe1997differential,Barakat2004}, whenever the map
\[
  \text{Ad}:K\to GL\left(\kf\oplus\mR^m/\kf\right)=GL\left(m\right)
\]
is injective, $R_\zeta$ becomes a $K$-structure, namely, it is a $K$-subbundle of the frame bundle $LM$; in this case, the previous diagram has an analogous diagram at principal bundle level
\[
  \begin{tikzcd}[ampersand replacement=\&,row sep=1.7cm,column sep=1cm]
    M
    \&
    LM
    \arrow[hook]{d}{\gamma}
    \arrow[dl,hook,dashed]{}{}
    \arrow{l}{\tau}
    \&
    R_\zeta
    \arrow[hook]{d}{}
    \arrow[hook']{l}{}
    \\
    AM
    \arrow[equal]{r}{}
    \arrow{u}{\beta\circ\tau}
    \&
    LM\left[GL\left(m\right)\oplus\mR^m\right]
    \&
    R_\zeta\left[K\oplus\mR^m\right]
    \arrow[hook']{l}{}
  \end{tikzcd}      
\]
Here $AM$ indicates the affine frame bundle of the spacetime $M$, maps $\beta:AM\to LM,\gamma:LM\to AM$ are canonical maps between these bundles (see Appendix \ref{sec:lift-conn-affine}), and the symbol $P\left[G\right]$, where $\pi:P\to M$ is an $H$-principal bundle and $G\supset H$ is a Lie group containing $H$, indicates the extension of $P$ by enlarging its structure group to $G$ (this construction is detailed in Section \ref{sec:cart-conn-jet} below). Therefore, we will restrict ourselves to Cartan connections describing Palatini gravity on a $K$-structure $R_\zeta$; in this setting we devise a method to extend them to Cartan connections on $LM$, in order to reproduce the correspondence between Chern-Simons field theory and Palatini gravity in a case where the structure group is the general linear group $GL\left(m\right)$.

\par Let us briefly describe the structure of the article. The geometrical tools used throughout the article are presented in Section \ref{sec:some-geom-tools}. The description of Cartan connections as sections of a bundle is carried out in Section \ref{sec:cart-conn-jet}; this description is necessary due to the type of geometrical formulation adopted for the variational problems. The operations of extension and reduction for generalized Cartan connections are developed in Section \ref{sec:extens-gener-cart}. Although the author's knowledge of Cartan connections is far from exhaustive, it seems that these operations, even though they are a direct consequence of the procedure of identification between Cartan connections and principal connections on the extended bundle, have not been described previously in the literature. If this is indeed the case, this section represents an original contribution of the present article. The variational problems for Chern-Simons and gravity are described in Section \ref{sec:vari-probl-chern}. The main contribution of this section is to find a global formulation for the Chern-Simons field theory in terms of jet bundles (Section \ref{sec:wise-vari-probl}); the price that must be paid for this formulation is the appearance of an additional constraint (see Remark \ref{rem:Kk0Constraint}). The main result of the article is discussed in Section \ref{sec:extens-gener-cart-1}: A variational problem of the Chern-Simons type has been found, but with a structure group given by the general affine group (Section \ref{sec:chern-simons-vari}) such that the extremals of any Chern-Simons theory (described in terms of the Section \ref{sec:wise-vari-probl}) are in a bijective correspondence with their extremals through the extension and reduction operations of Section \ref{sec:extens-gener-cart}.

\subsubsection*{Notations} We are adopting here the notational conventions from \cite{saunders89:_geomet_jet_bundl} when dealing with bundles and its associated jet spaces. It means that, given a bundle $\pi:P\to M$, there exists a family of bundles and maps fitting in the following diagram
\[
  \begin{tikzcd}[ampersand replacement=\&,row sep=.5cm,column sep=1cm]
    \cdots
    \arrow{r}{}
    \&
    J^{k+1}\pi
    \arrow{r}{\pi_{k+1,k}}
    \arrow[swap]{rrdd}{\pi_{k+1}}
    \&
    J^{k}\pi
    \arrow{r}{\pi_{k,k-1}}
    \arrow{rdd}{\pi_{k}}
    \&
    \cdots
    \arrow{r}{\pi_{21}}
    \&
    J^1\pi
    \arrow[swap]{ddl}{\pi_1}
    \arrow{r}{\pi_{10}}
    \&
    P
    \arrow{lldd}{\pi}
    \\
    \&
    \cdots
    \&
    \&
    \cdots
    \&
    \&
    \\
    \&
    \&
    \&
    M
    \&
    \&
  \end{tikzcd}  
\]
Sections of $\pi:P\to M$ will be indicate by the symbol $\Gamma\pi$. The set of vectors tangent to $P$ in the kernel of $T\pi$ will be represented with the symbol $V\pi\subset TP$. In this regard, the set of vector fields which are vertical for a bundle map $\pi:P\to M$ will be indicated by $\mathfrak{X}^{V\pi}\left(P\right)$. The space of differential $p$-forms, sections of $\Lambda^p (T^*Q)\to Q$, will be denoted by $\Omega^p(Q)$. {We also write $\Lambda^\bullet(Q)=\bigoplus_{j=1}^{\dim Q}\Lambda^j(T^*Q)$}. If $f\colon P\to Q$ is a smooth map and $\alpha_x$ is a $p$-covector on $Q$, we will sometimes use the notation $\alpha_{f(x)}\circ T_xf$ to denote its pullback $f^*\alpha_x$. If $P_1\to Q$ and $P_2\to Q$ are fiber bundles over the same base $Q$ we will write $P_1\times_Q P_2$ for their fibered product, or simply $P_1\times P_2$ if there is no risk of confusion. Unless explicitly stated, the canonical projections onto its factor will be indicated by
\[
  \text{pr}_i:P_1\times P_2\to P_i,\qquad i=1,2.
\]
Given a manifold $N$ and a Lie group $G$ acting on $N$, the symbol $\left[n\right]_G$ for $n\in N$ will indicate the $G$-orbit in $N$ containing $n$; the canonical projection onto its quotient will be denoted by
\[
  p_G^N:N\to N/G.
\]
Also, if $\mathfrak{g}$ is the Lie algebra for the group $G$, the symbol $\xi_N$ will represent the infinitesimal generator for the $G$-action asssociated to $\xi\in\mathfrak{g}$. Finally, Einstein summation convention will be used everywhere.

We will be working with several different types of connections, so it will be necessary to have a notational convention for them. Giving a $H$-principal bundle $p:Q\to M$, the first jet bundle $J^1p$ has a $\hf$-valued $1$-form that is called \emph{canonical connection}; this form will be indicated by the symbol
\[
  \theta_{J^1p}\in\Omega^1\left(J^1p,\hf\right).
\]
Because it will become a connection form, its associated curvature form will be denoted by
\[
  \Theta_{J^1p}\in\Omega^2\left(J^1p,\hf\right).
\]
On the other hand, two different types of connection will be considered on a principal bundle $p:Q\to M$. First, it could have a connection form, which will be indicated as $\omega_Q$; its asscoiated curvature form will be $\Omega_Q$. Moreover, if the structure group $H$ is a subgroup of a bigger Lie group $G\supset H$, then we can have on $Q$ a $\g$-valued $1$-form, which will be denoted as $A_Q$, and the curvature connection will become $F_Q$.

\section{Some geometrical tools}
\label{sec:some-geom-tools}

The following section collects geometrical facts regarding principal bundles and Cartan connections on them. 

\subsection{Geometry of the jet space of a principal bundle}
\label{sec:geometry-jet-bundle}

Throughout the article, we will make extensive use of the geometrical tools related to the jet space associated to a principal bundle, as well as its connection bundle, as they are discussed in \cite{MR0315624,springerlink:10.1007/PL00004852}. So, in order to proceed, let $p:Q\to N$ be a principal bundle with structure group $H$; then we can lift the right action of $H$ on $Q$ to a $H$-action on $J^1p$, and so define the bundle
\[
  \overline{p}:C\left(Q\right):=J^1p/H\to M
\]
fitting in the following diagram
\[
  \begin{tikzcd}[ampersand replacement=\&,row sep=.7cm,column sep=.7cm]
    \&
    J^1p
    \arrow[swap]{dl}{p_{10}}
    \arrow{dr}{p_H^{J^1p}}
    \&
    \\
    Q
    \arrow[swap]{dr}{p}
    \&
    \&
    C\left(Q\right)
    \arrow{dl}{\overline{p}}
    \\
    \&
    M
    \&
  \end{tikzcd}  
\]
It can be proved \cite{springerlink:10.1007/PL00004852} that this diagram defines $J^1p$ as a pullback, namely, that
\begin{equation}\label{eq:JetAsConnection}
  J^1p=p^*Q=Q\times_M C\left(Q\right).
\end{equation}
We know that $J^1p$ comes equipped with the \emph{contact structure}, that can be described by means of a $Vp$-valued $1$-form
\begin{equation}\label{eq:CanonicalConnection}
  \left.\theta_{J^1p}\right|_{j_x^1s}:=T_{j_x^1s}p_{10}-T_xs\circ T_{j_x^1s}p_1;
\end{equation}
moreover, because $p:Q\to N$ is a principal bundle, we have the bundle isomorphism on $Q$
\[
  Vp\simeq Q\times\hf.
\]
It means that we can consider $\theta$ as a $\hf$-valued $1$-form; in fact, with respect to the $K$-principal bundle structure
\[
  p_H^{J^1p}:J^1p\to C\left(Q\right),
\]
the $1$-form $\theta$ becomes a connection form, dubbed \emph{canonical connection form}. It has the following property.

\begin{proposition}
  Let $\Gamma_Q:Q\to J^1p$ be a connection on $Q$. Then its connection form $\omega_Q\in\Omega^1\left(Q,\hf\right)$ can be constructed from the canonical connection form through pullback along $\Gamma_Q$,
  \[
    \omega_Q=\Gamma_Q^*\theta_{J^1p}.
  \]
\end{proposition}

For any manifold $M$ of dimension $m$, the bundle $\tau:LM\to M$ is the \emph{frame bundle of $M$}, is defined through
\begin{equation}\label{eq:FrameBundle}
  LM:=\bigcup_{x\in M}\left\{u:\mR^m\to T_xM\text{ linear and with inverse}\right\}.
\end{equation}
It has a canonical free $GL\left(m,\mR\right)$-action, given by the formula 
\[
  u\cdot A:=u\circ A,\qquad u\in LM,A\in GL\left(m,\mR\right).
\]

Let $\eta\in M_m\left(\mR\right)$ be a non degenerate symmetric $m\times m$-matrix with real entries; for definiteness, we will fix
\[
  \eta:=
  \begin{bmatrix}
    -1&0&\cdots&0\\
    0&1&&0\\
    \vdots&&\ddots&\vdots\\
    0&\cdots&0&1
  \end{bmatrix},
\]
although the constructions we will consider in the present article should work with any signature. Then we have a Lie group $K\subset GL\left(m,\mR\right)$ defined by
\begin{equation}\label{eq:LorentzGroup}
  K:=\left\{A\in M_m\left(\mR\right):A\eta A^T=\eta\right\}.
\end{equation}
Then we have an action of $K$ on $LM$; it yields to a bundle
\[
  \tau_\Sigma:\Sigma:=LM/K\to M.
\]
\begin{lemma}
  The bundle $\Sigma$ is the bundle of metrics of $\eta$-signature on $M$.
\end{lemma}
Let us indicate with $\kf$ the Lie algebra of $K$; then, we have that
\[
  \kf:=\left\{a\in\mathfrak{gl}\left(m,\mR\right):a\eta+\eta a^T=0\right\}.
\]
Accordingly, we can define the $K$-invariant subspace
\begin{equation}\label{eq:Transvections}
  \pf:=\left\{a\in\mathfrak{gl}\left(m,\mR\right):a\eta-\eta a^T=0\right\},
\end{equation}
usually called \emph{transvections} (see \cite{wise2009symmetric}); it follows that
\[
  \mathfrak{gl}\left(m,\mR\right)=\kf\oplus\pf.
\]

We will need this result concerning some natural properties of the canonical connections.

\begin{lemma}\label{lem:CanonicalConnectionNatural}
  Let $p:Q\to N$ be a $H$-principal bundle and $i_\zeta:R_\zeta\hookrightarrow Q$ a $L$-principal subbundle, with projection
  \[
    p_\zeta:R_\zeta\to N.
  \]
  We have the following relation between the canonical connections on $J^1p$ and $J^1p_\zeta$, namely
  \[
    \left(j^1i_\zeta\right)^*\theta_{J^1p}=\theta_{J^1p_\zeta},
  \]
  where $\theta_{J^1p}\in\Omega^1\left(J^1p,\mathfrak{h}\right)$ and $\theta_{J^1p_\zeta}\in\Omega^1\left(J^1p_\zeta,\mathfrak{l}\right)$ are the corresponding canonical connection forms.
\end{lemma}
\begin{proof}
  The lemma follows from the formula
  \[
    \left.\theta_{J^1p}\right|_{j_x^1s}\circ T_{j_x^1\sigma}\left(j^1i_\zeta\right)=T_{\sigma\left(x\right)}i_\zeta\circ\left.\theta_{J^1p_\zeta}\right|_{j_x^1\sigma},
  \]
  valid for any $j_x^1s=j^1i_\zeta\left(j_x^1\sigma\right)$ and $j_x^1\sigma\in J^1p_\zeta$.
\end{proof}

Assuming that some topological conditions on the manifold $M$ hold\footnote{The existence of a metric with $\left(p,q\right)$-signature is equivalent to the splitting of the tangent bundle $TM$ in a direct sum of vector subbundles of rank $p$ and $q$ respectively.}, a family
\[
  \left\{O_\zeta:\zeta:M\to\Sigma\text{ metric}\right\}
\]
of subbundles of $LM$ can be constructed; namely, let us define
\[
  O_\zeta:=\left\{u\in LM:\zeta\left(u\left(w_1\right),u\left(w_2\right)\right)=\eta\left(w_1,w_2\right)\text{ for all }w_1,w_2\in\mR^m\right\}.
\]
Here $\eta:\mR^m\times\mR^m\to\mR$ indicates the bilinear form associated to the matrix $\eta$.


\subsection{Introduction to Cartan geometry}
\label{sec:intr-cart-geom}

Let us recall the definition of Cartan connections on a principal bundle \cite{Wise_2010,alekseevsky95:_differ_cartan,sharpe1997differential}. Let $\left(G,H\right)$ be a pair of Lie groups such that $H\subset G$ is a closed subgroup and $G/H$ is connected. Recall that for every principal bundle $\pi:P\to M$ with structure group $H$ we have a map
\[
  \kappa_H:VP\to\hf
\]
such that
\[
  \kappa_H\left(\xi_P\left(u\right)\right)=\xi
\]
for every $\xi\in\hf$ and every $u\in P$.

\begin{definition}[Cartan geometry]
  A \emph{Cartan geometry} modelled on the pair $\left(G,H\right)$ is an $H$-principal bundle $\pi:P\to M$ together with a $\g$-valued $1$-form $A$ on $P$ such that
  \begin{enumerate}[(i)]
  \item $A_u:T_uP\to\g$ is a linear isomorphism for every $u\in P$,
  \item $R_h^*A=\text{Ad}_{h^{-1}}\circ A$ for every $h\in H$, and
  \item $\left.A\right|VP=\kappa_H$.
  \end{enumerate}
  The form $A$ is called the \emph{(principal) Cartan connection} for the given Cartan geometry.
\end{definition}

A Cartan geometry on a manifold provides us with a associated bundle description of its tangent bundle, as the following result indicates.

\begin{proposition}
  Let $\left(\pi:P\to M,A\right)$ be a Cartan geometry modeled on the pair $\left(G,H\right)$. Then $TM$ is isomorphic as vector bundle to the associated bundle
  \[
    P\times_H\g/\hf,
  \]
  where $H$ acts on $\g/\hf$ through the quotient representation induced by the $\text{Ad}\left(H\right)$-action on $\g$.
\end{proposition}
\begin{proof}
  Let $u\in P$; because $A_u:T_uP\to\g$ satisfies $\left.A\right|VP=\kappa_H$, it induces a morphism
  \[
    \phi_u:T_xM\to\g/\hf:v_x\mapsto\pi_{\g/\hf}\left(A_u\left(\widehat{v_x}\right)\right),
  \]
  where $x=\pi\left(u\right)$, $\widehat{v_x}\in T_uP$ is any element projecting onto $v_x$ via $T_u\pi$, and
  \[
    \pi_{\g/\hf}:\g\to\g/\hf
  \]
  is the canonical projection onto the quotient. It becomes an isomorphism because $A_u$ is also an isomorphism, and given that this map is $H$-equivariant, we have that
  \[
    \phi_{u\cdot h}=\text{Ad}_{h^{-1}}\circ\phi_u
  \]
  for every $u\in P$ ad $h\in H$. Therefore, we have the isomorphism of bundles $\overline{q}:P\times_H\g/\hf\to TM$ given by
  \[
    \overline{q}\left(\left[u,\left[\xi\right]_\hf\right]_H\right):=\left(\phi_u\right)^{-1}\left(\left[\xi\right]_\hf\right),
  \]
  that can be proved to be well-defined.
\end{proof}

We should stress that a Cartan connection is not a principal connection on $P$, because it takes values in the larger Lie algebra $\g$; nevertheless, there exists a relationship between these concepts \cite{panak1998natural}. Consider $H$ as a $G$-space through left multiplication, and construct the associated bundle
\[
  P\left[G\right]:=P\times_HG;
\]
it is a $G$-principal bundle, and we have a canonical inclusion $\gamma:P\hookrightarrow P\left[G\right]$ given by
\[
  \gamma\left(u\right):=\left[u,e\right]_H.
\]
It is known that every morphism of $G$-principal bundles on a fixed basis $M$ is an isomorphism \cite{steenrod1999the}; this fact gives rise to the following property of the extension procedure described above.
\begin{lemma}\label{lem:BundleIsomorphism}
  Let $P$ be an $H$-principal bundle on $M$ and suppose further that it is a subbundle of a $G$-principal bundle $Q$ on $M$. Then $Q\simeq P\left[G\right]$.
\end{lemma}
\begin{proof}
  Let $i:P\hookrightarrow Q$ be the immersion of $P$ into $Q$; then we have the morphism
  \[
    \phi:P\left[G\right]\to Q:\left[u,g\right]_H\mapsto i\left(u\right)\cdot g.
  \]
  It is a morphism of $G$-principal bundles on $M$; therefore, $Q$ and $P\left[G\right]$ are isomorphic, as desired.
\end{proof}

The extension of a principal bundle can be used to relate Cartan connections with principal connections \cite{alekseevsky95:_differ_cartan,sharpe1997differential,panak1998natural}.

\begin{proposition}\label{prop:CartanVsPrincipalConn}
  The Cartan connection $A:TP\to\g$ induces on $P\left[G\right]$ a unique principal connection form $\widetilde{A}:TP\left[G\right]\to\g$ such that
  \[
    \gamma^*\widetilde{A}=A.
  \]
  Conversely, suppose that $\dim P=\dim G$ and let $\widetilde{A}$ be a principal connection on $P\left[G\right]$ such that
  \begin{equation}\label{eq:CrossCondition1}
    \ker\widetilde{A}\cap T\gamma\left(TP\right)=\left\{0\right\}.
  \end{equation}
  Then $A:=\gamma^*\widetilde{A}$ is a Cartan connection on $P$.
\end{proposition}
\begin{proof}
  Because the quotient map $p_{H}^{P\times G}:P\times G\to P\left[G\right]$ is surjective, any element $W\in T_{\left[u,g\right]_H}P\left[G\right]$ can be represented as
  \[
    W=T_{\left(u,g\right)}p_{H}^{P\times G}\left(X_u,T_eL_g\zeta\right)
  \]
  for $X_u\in T_uP$ and $\zeta\in\g$. Then we define
  \[
    \left.\widetilde{A}\right|_{\left[u,g\right]_H}\left(W\right):=\zeta+\text{Ad}_{g^{-1}}\left.A\right|_u\left(X_u\right).
  \]
  It can be proved that it is well-defined, and defines a principal connection on $P\left[G\right]$.

  Now, let us suppose that we have a principal connection $\widetilde{A}$ on $P\left[G\right]$ and define
  \[
    A:=\gamma^*\widetilde{A}:TP\to\g.
  \]
  Let us verify that it is a Cartan connection on $P$. Because of the condition \eqref{eq:CrossCondition1}, the map
  \[
    A_u:T_uP\to\g
  \]
  is a monomorphism; because $\dim{T_uP}=\dim P=\dim G=\dim\g$, we have that this map is an isomorphism.

  Aditionally, for any $h\in H$
  \begin{align*}
    R_h^*A&=R_h^*\gamma^*\widetilde{A}\\
          &=\left(\gamma\circ R_h\right)^*\widetilde{A}\\
          &=\left(R_h\circ\gamma\right)^*\widetilde{A}\\
          &=\gamma^*R_h^*\widetilde{A}\\
          &=\gamma^*\left(\text{Ad}_{h^{-1}}\circ\widetilde{A}\right)\\
          &=\text{Ad}_{h^{-1}}\circ A.
  \end{align*}
  Finally, let us take $\xi\in\hf$; because of the identity
  \[
    T_u\gamma\left(\xi_P\left(u\right)\right)=\xi_{P\left[G\right]}\left(\gamma\left(u\right)\right)
  \]
  we will have that
  \[
    A_u\left(\xi_P\left(u\right)\right)=\widetilde{A}_{\gamma\left(u\right)}\left(T_u\gamma\left(\xi_P\left(u\right)\right)\right)=\widetilde{A}_{\gamma\left(u\right)}\left(\xi_{P\left[G\right]}\left(\gamma\left(u\right)\right)\right)=\xi
  \]
  for any $u\in P$.
\end{proof}


\section{Cartan connections and jet bundles}
\label{sec:cart-conn-jet}

The basic idea for the geometrical interpretation of the correspondence between Chern-Simons field theory and gravity is due to Wise \cite{wise2009symmetric,Wise_2010}, and uses a Chern-Simons Lagrangian evaluated on forms that are not principal connections, but Cartan connections. As we have pointed out before, this approach is not convenient when you try to understand Chern-Simons field theory from the viewpoint of geometric mechanics (i.e., the setting described in \cite{Gotay:1997eg,Gotay:2004ib,Gotay1991203,MarkJ1991375,PROP:PROP2190440304}). Therefore, we will devote the next section to translate the formalism of Cartan connections into the realm of jet bundles, in order to have at our disposal a language suitable for the description of Chern-Simons field theory from this viewpoint.

\subsection{Canonical (generalized) Cartan connection on a jet bundle}
\label{sec:canon-gener-cart}

Thus, we have that a Cartan connection on an $H$-principal bundle $P$ can be seen as a principal connection $\widetilde{A}$ on the extended bundle $P\left[G\right]$, provided that
\begin{enumerate}[(a)]
\item\label{item:DimensionCartan} $P$ has the same dimension than $G$, and
\item\label{item:ComplementaryCartan} the horizontal spaces of this connection are complementary to the tangent spaces of $P$ (viewed as subspaces of the tangent spaces of $P\left[G\right]$).
\end{enumerate}

We can reformulate the second condition in terms of jet bundles; in fact, let us define the set
  \[
    U_\gamma:=\left\{j_x^1s\in J^1\pi_{\left[G\right]}:T_{s\left(x\right)}R_{g^{-1}}\left(T_xs\left(T_xM\right)\right)\cap T_u\gamma\left(T_uP\right)=\left\{0\right\}\text{ iff }s\left(x\right)=\left[u,g\right]\right\}.
  \]
Using decomposition \eqref{eq:JetAsConnection}, we can see that sections having its images in this set corresponds exactly with connections on $P\left[G\right]$ that verify condition \ref{item:ComplementaryCartan} above. It gives us the following characterization for condition \eqref{eq:CrossCondition1} in terms of the jet bundle of the extended principal bundle $P\left[G\right]$.
\begin{proposition}
A connection $\widetilde{\Gamma}:P\left[G\right]\to J^1\pi_{\left[G\right]}$ will satisfy Equation \eqref{eq:CrossCondition1} if and only if $\widetilde{\Gamma}\left(\left[u,g\right]\right)\in U_\gamma$ for all $\left[u,g\right]\in P\left[G\right]$.
\end{proposition}
Now, we have the following result.
\begin{lemma}
  The set $U_\gamma$ is open in $J^1\pi_{\left[G\right]}$.
\end{lemma}
\begin{proof}
  Given $V_1,V_2\subset V$ a pair of subspaces of the (finite-dimensional) vector space $V$, fix basis $B_1,B_2$ and $B$ for each of them. Thus the condition
  \[
    V_1\cap V_2=\left\{0\right\}
  \]
  is equivalent to $F_k\not=0$, where $F_k$ is the sum of the squares of the $k\times k$ minors of the matrix formed by the components of the vectors in $B_1\cup B_2$ respect to the basis $B$, where $k=\dim V_1+\dim V_2$. Consider now
  \begin{align*}
    V&=T_{s\left(x\right)}P\left[G\right]\\
    V_1&=T_{s\left(x\right)}R_{g^{-1}}\left(T_xs\left(T_xM\right)\right)\\
    V_2&=T_u\gamma\left(T_uP\right);
  \end{align*}
  in terms of local coordinates on $J^1\pi_{\left[G\right]}$, $F_k$ gives rise to a polynomial in the jet variables, and the result follows. 
\end{proof}
Therefore, instead of considering the variational problem for a Cartan connection over the entire bundle $J^1\pi_{\left[G\right]}$, we can restrict ourselves to the open set $U_\gamma$, thereby abandoning condition \ref{item:ComplementaryCartan} above. With this consideration in mind, let us generalize the notion of Cartan connection; to this end we will use the notion of \emph{generalized Cartan connection}, as it is defined in \cite{alekseevsky95:_differ_cartan}.

\begin{definition}[Generalized Cartan connection]
  A \emph{generalized (principal) $\g/\hf$-Cartan connection} on an $H$-principal bundle $\pi:P\to M$ is a $\g$-valued $1$-form $A$ on $P$ such that 
  \begin{enumerate}[(i)]
  \item $R_h^*A=\text{Ad}_{h^{-1}}\circ A$ for every $h\in H$, and
  \item $\left.A\right|V\pi=\kappa_H$.
  \end{enumerate}
\end{definition}

Let us consider the following diagram 
\[
  \begin{tikzcd}[ampersand replacement=\&,row sep=1.7cm,column sep=2.5cm]
    J^1\left(\pi\circ\text{pr}_1\right)
    \arrow{r}{\left(\pi\circ\text{pr}_1\right)_{10}}
    \arrow[swap]{d}{j^1p^{P\times G}_H}
    \&
    P\times G
    \arrow{d}{p^{P\times G}_H}
    \arrow{r}{\text{pr}_1}
    \&
    P\arrow{d}{\pi}
    \\
    J^1\pi_{\left[G\right]}
    \arrow{r}{\left(\pi_{\left[G\right]}\right)_{10}}
    \&
    P\left[G\right]
    \arrow{r}{\pi_{\left[G\right]}}
    \&
    M
  \end{tikzcd}      
\]
The bundle $\pi\circ\text{pr}_1:P\times G\to M$ together with the action
\[
  \left(u,g\right)\cdot\left(h,g'\right):=\left(u\cdot h,gg'\right)
\]
is an $H\times G$-principal bundle. Thus we have the following result.

\begin{lemma}\label{lem:ProductConnection}
  If $\theta_{J^1\left(\pi\circ\text{pr}_1\right)}\in\Omega^1\left(J^1\left(\pi\circ\text{pr}_1\right),\hf\times\g\right)$ denotes the canonical connection on $J^1\left(\pi\circ\text{pr}_1\right)$ and $\theta_{J^1\pi_{\left[G\right]}}\in\Omega^1\left(J^1\pi_{\left[G\right]},\g\right)$ is the canonical connection on $J^1\pi_{\left[G\right]}$, then
  \[
    \left(j^1p_H^{P\times G}\right)^*\theta_{J^1\pi_{\left[G\right]}}=\text{pr}_2^{\hf\times\g}\circ\theta_{J^1\left(\pi\circ\text{pr}_1\right)},
  \]
  where
  \[
    \text{pr}_2^{\hf\times\g}:\hf\times\g\to\g
  \]
  stands for the projection onto the second factor.
\end{lemma}
\begin{proof}
  Recall that the canonical connection on the jet bundle of a principal bundle $\pi:P\to M$ is nothing but the $V\pi$-valued contact form, which can be seen as $\g$-valued through the identification
  \[
    V\pi\simeq P\times\g.
  \]
  It means that they have the property
  \[
    \left(j^1p_H^{P\times G}\right)^*\theta_{J^1\pi_{\left[G\right]}}=Tp^{P\times G}_H\circ\theta_{J^1\left(\pi\circ\text{pr}_1\right)},
  \]
  and the lemma follows from the identification between the vertical bundles and the Lie algebras.
\end{proof}

\subsection{Generalized Cartan connections as equivariant sections of a jet bundle}
\label{sec:cart-conn-as}

Recall that a principal connection on $P$ can be represented by an equivariant map $\Gamma:P\to J^1\pi$. The next result gives an analogous representation for generalized Cartan connections.

\begin{proposition}\label{prop:GeneralizedCartanConnection}
  Any generalized Cartan connection $A:TP\to\g$ gives rise to a $H$-equivariant bundle map
  \[
    \Gamma_A:P\to J^1\pi_{\left[G\right]}
  \]
  covering the immersion $\gamma:P\hookrightarrow P\left[G\right]$. Conversely, for any $H$-equivariant bundle map $\Gamma:P\to J^1\pi_{\left[G\right]}$ making the following diagram
  \[
    \begin{tikzcd}[ampersand replacement=\&,row sep=1.7cm,column sep=1.5cm]
      \&
      J^1\pi_{\left[G\right]}
      \arrow{d}{\left(\pi_{\left[G\right]}\right)_{10}}
      \\
      P
      \arrow{ur}{\Gamma}
      \arrow[hook]{r}{\gamma}
      \&
      P\left[G\right]
    \end{tikzcd}      
  \]
  commutative, there exists a generalized Cartan connection $A_\Gamma:TP\to\g$ such that
  \[
    A_\Gamma=\Gamma^*\theta_{J^1\pi_{\left[G\right]}}.
  \] 
\end{proposition}
\begin{proof}
  The map $\Gamma_A:P\to J^1\pi_{\left[G\right]}$ is given by
  \[
    \Gamma_A\left(u\right):v\in T_xM\mapsto T_{\left(u,e\right)}p_H^{P\times G}\left(\widehat{v}_u,-A_u\left(\widehat{v}_u\right)\right),
  \]
  where $\widehat{v}_u\in T_uP$ is any lifting of $v\in T_xM$ to $T_uP$, and
  \[
    p_H^{P\times G}:P\times G\to P\left[G\right]
  \]
  stands for the canonical projection onto the quotient; it is well-defined because of the identity
  \[
    T_{\left(u,g\right)}p_H^{P\times G}\left(-\zeta_P\left(u\right),T_eR_g\left(\zeta\right)\right)=0
  \]
  for every $\left(u,g\right)\in P\times G$ and $\zeta\in\hf$.

  Let us now prove that
  \[
    \Gamma_A^*\theta_{J^1\pi_{\left[G\right]}}=A
  \]
  where $\theta_{J^1\pi_{\left[G\right]}}\in\Omega^1\left(J^1\pi_{\left[G\right]},\g\right)$ is the canonical connection form on $J^1\pi_{\left[G\right]}$. In order to accomplish it, let us define the section $\widetilde{\Gamma}:P\left[G\right]\to J^1\pi_{\left[G\right]}$ for $\left(\pi_{\left[G\right]}\right)_{10}$ such that
  \[
    \widetilde{\Gamma}\left(\left[u,g\right]\right):v\in T_xM\mapsto T_{\gamma\left(u\right)}R_g\left(\Gamma_A\left(u\right)\left(v\right)\right)\in T_{\left[u,g\right]}P\left[G\right].
  \]
  Accordingly, we have the following commutative diagram
    \[
    \begin{tikzcd}[ampersand replacement=\&,row sep=1.7cm,column sep=1.5cm]
      \&
      J^1\pi_{\left[G\right]}
      \\
      P
      \arrow{ur}{\Gamma_A}
      \arrow[hook]{r}{\gamma}
      \&
      P\left[G\right]
      \arrow[swap]{u}{\widetilde{\Gamma}}
    \end{tikzcd}      
  \]
  Moreover, by fixing an auxiliary connection $\omega_0\in\Omega^1\left(P,\hf\right)$ for $P$, we can define a map $\widehat{\Gamma}_A^{\omega_0}:P\times G\to J^1\left(\pi\circ\text{pr}_1\right)$ such that
  \[
    \widehat{\Gamma}_A^{\omega_0}\left(u,g\right)\left(v\right):=\left(v_u^H,-T_eR_{g^{-1}}\left(A_u\left(v^H_u\right)\right)\right)
  \]
  for every $\left(u,g\right)\in P\times G$ and $v\in T_xM$; here $v_u^H$ indicates the horizontal lifting of $v$ to $T_uP$ by means of the connection $\omega_0$. These maps fit in the following diagram
  \[
    \begin{tikzcd}[ampersand replacement=\&,row sep=1.7cm,column sep=2.5cm]
      J^1\left(\pi\circ\text{pr}_1\right)
      \arrow[swap]{d}{j^1p^{P\times G}_H}
      \&
      P\times G
      \arrow[swap]{l}{\widehat{\Gamma}^{\omega_0}_A}
      \arrow[swap]{d}{p^{P\times G}_H}
      \&
      P
      \arrow[hook',swap]{l}{\text{inc}}
      \arrow[hook]{ld}{\gamma}
      \\
      J^1\pi_{\left[G\right]}
      \&
      P\left[G\right]
      \arrow[swap]{l}{\widetilde{\Gamma}}
      \&
    \end{tikzcd}
  \]
  Namely, Cartan connection $\Gamma_A$ can be retrieved through the formula
  \[
    \Gamma_A=\widetilde{\Gamma}\circ\gamma=j^1p_H^{P\times G}\circ\widehat{\Gamma}^{\omega_0}_A\circ\text{inc};
  \]
  therefore, because
  \[
    \left(\widehat{\Gamma}^{\omega_0}_A\circ\text{inc}\right)^*\theta_{J^1\left(\pi\circ\text{pr}_1\right)}=\omega_0+A,
  \]
  Lemma \ref{lem:ProductConnection} implies that
  \[
    \Gamma_A^*\theta_{J^1\pi_{\left[G\right]}}=A
  \]
  as required.
\end{proof}

We can rephrase Proposition \ref{prop:CartanVsPrincipalConn} using this correspondence.
\begin{corollary}
  For any map $\Gamma:P\to J^1\pi_{\left[G\right]}$ covering $\gamma$ we have a section $\widetilde{\Gamma}:P\left[G\right]\to J^1\pi_{\left[G\right]}$ for $\left(\pi_{\left[G\right]}\right)_{10}$, which is defined through
  \[
    \widetilde{\Gamma}\left(\left[u,g\right]\right):v\in T_xM\mapsto T_{\gamma\left(u\right)}R_g\left(\Gamma\left(u\right)\left(v\right)\right)\in T_{\left[u,g\right]}P\left[G\right].
  \]
  Conversely, any connection $\widetilde{\Gamma}:P\left[G\right]\to J^1\pi_{\left[G\right]}$ gives rise to a map $\Gamma:P\to J^1\pi_{\left[G\right]}$ covering $\gamma$ by restriction to $\gamma\left(P\right)\subset P\left[G\right]$.
\end{corollary}

Recall also the usual extension of principal connections in this context.

\begin{proposition}
  Let $H\subset G$ be a pair of Lie groups. Consider a $G$-principal bundle $\tau:Q\to M$ and let $\pi:P\to M$ be a $H$-principal subbundle; let $i:P\hookrightarrow Q$ be the canonical immersion. For every connection
  \[
    \Gamma:P\to J^1\pi
  \]
  the extension $\widehat{\Gamma}:Q\to J^1\tau$ of $\Gamma$ is the map
  \[
    \widehat{\Gamma}\left(\widehat{u}\right):=j^1i\left(j^1R_g\left(\Gamma\left(u\right)\right)\right)
  \]
  if and only if $\widehat{u}=i\left(u\right)g$.
\end{proposition}

\subsection{The bundle of (generalized) Cartan connections}
\label{sec:bundle-cart-conn}

We know \cite{MR0315624,springerlink:10.1007/PL00004852} that connections on a $H$-principal bundle $\pi:P\to M$ can be seen either as equivariant sections of the map $\pi_{10}:J^1\pi\to P$ or as sections of the quotient bundle
\[
  \overline{\pi}:J^1\pi/H\to M.
\]
Both descriptions are related by the following result.
\begin{proposition}
  Let $\pi:P\to M$ be an $H$-principal bundle and suppose that we have a bundle $p:Q\to P$ together with an $H$-action on $Q$ such that $p$ is an $H$-equivariant map; consider the induced bundle $\overline{p}:Q/H\to M$. Then, any section $\sigma:M\to Q/H$ of $\overline{p}$ can be lifted to a section $\widehat{\sigma}:P\to Q$ of $p$ and viceversa, any section of $p$ induces a section of $\overline{p}$ by quotient.
\end{proposition}
\begin{proof}
  Because the $H$-action on $P$ is free, given $\left[q\right]\in Q/H$ and $u\in\pi^{-1}\left(\overline{p}\left(\left[q\right]\right)\right)$, there exists a unique element $\widetilde{q}\in\left[q\right]$ such that
  \[
    p\left(\widetilde{q}\right)=u.
  \]
  Thus, for every $x\in M$, define
  \[
    \widehat{\sigma}\left(u\right):=\widetilde{q}
  \]
  if and only if $\sigma\left(x\right)=\left[q\right]$ and $\widetilde{q}\in\left[q\right]$ is such that
  \[
    p\left(\widetilde{q}\right)=u.\qedhere
  \]
\end{proof}

In the case of Cartan connections, the description in terms of equivariant maps is provided by the following corollary of Proposition \ref{prop:GeneralizedCartanConnection}.
\begin{corollary}\label{cor:GeneralizedCartan}
  A generalized Cartan connection can be represented as an $H$-equivariant section of the pullback bundle
  \[
    \begin{tikzcd}[ampersand replacement=\&,row sep=1.7cm,column sep=1.5cm]
      \gamma^*\left(J^1\pi_{\left[G\right]}\right)
      \arrow{r}{\text{pr}_2^\gamma}
      \arrow{d}{\text{pr}_1^\gamma}
      \&
      J^1\pi_{\left[G\right]}
      \arrow{d}{\left(\pi_{\left[G\right]}\right)_{10}}
      \\
      P
      \arrow[hook]{r}{\gamma}
      \&
      P\left[G\right]
    \end{tikzcd}      
  \]
\end{corollary}
Thus, for $\g/\hf$-Cartan connections, we have the following commutative diagram of bundles
\[
  \begin{tikzcd}[ampersand replacement=\&,row sep=1.7cm,column sep=1.5cm]
    \gamma^*\left(J^1\pi_{\left[G\right]}\right)
    \arrow{r}{\text{pr}_1^\gamma}
    \arrow[swap]{d}{p_H^{\gamma^*\left(J^1\pi_{\left[G\right]}\right)}}
    \&
    P
    \arrow{d}{\pi}
    \\
    \gamma^*\left(J^1\pi_{\left[G\right]}\right)/H
    \arrow{r}{\overline{\text{pr}}_1^\gamma}
    \&
    M
  \end{tikzcd}      
\]
where on $\gamma^*\left(J^1\pi_{\left[G\right]}\right)$ the $H$-diagonal action is considered. By proceeding in analogy with the principal connections case, we obtain the following definition.

\begin{definition}[Bundle of (generalized) $\g/\hf$-Cartan connections]
  The \emph{bundle of (generalized) $\g/\hf$-Cartan connections} is the bundle
  \[
    \overline{\text{pr}}_1^\gamma:\gamma^*\left(J^1\pi_{\left[G\right]}\right)/H\to M.
  \]
\end{definition}

\begin{remark}\label{rem:CaartanConnectionAndCanonicalConnection}
  The Cartan connection form is thus recovered from a section
  \[
    \sigma:M\to\gamma^*\left(J^1\pi_{\left[G\right]}\right)/H
  \]
  through the following procedure: we construct the unique $H$-equivariant section $\widehat{\sigma}:P\to\gamma^*\left(J^1\pi_{\left[G\right]}\right)$, which can be seen as a map $\Gamma_\sigma:P\to J^1\pi_{\left[G\right]}$. The Cartan connection form is then the pullback form
  \[
    A_\sigma:=\Gamma_\sigma^*\theta_{J^1\pi_{\left[G\right]}}\in\Omega^1\left(P,\g\right).
  \]
  Conversely, given a Cartan connection form $A$, we use Proposition \ref{prop:GeneralizedCartanConnection} in order to construct an $H$-equivariant map $\Gamma_A:P\to J^1\pi_{\left[G\right]}$, and so a section
  \[
    \sigma_A:M\to\gamma^*\left(J^1\pi_{\left[G\right]}\right)/H,
  \]
  as required.
\end{remark}



\section{Extensions and reductions of generalized Cartan connections}
\label{sec:extens-gener-cart}
We want to describe the connection between gravity and field theory using generalized connections on principal fiber bundles with the general linear group as the structure group. To accomplish this task, it will be essential to have an operation of extension and reduction of connections similar to those available for principal connections, but that work on generalized Cartan connections.

Therefore, in the present section we will use the correspondence between generalized Cartan connections and principal connections on a extended bundle, as described by Proposition \ref{prop:CartanVsPrincipalConn}, in order to generalize these constructions to the realm of (generalized) Cartan connection.

\subsection{How to extend a generalized Cartan connection}
\label{sec:how-extend-gener}

Let us consider the following problem: Given the diagram of Lie groups
\begin{equation}\label{eq:GroupsDiagram}
  \begin{tikzcd}[ampersand replacement=\&,row sep=.5cm,column sep=.7cm]
    \&
    G
    \&
    \\
    G_1
    \arrow[hook]{ur}{}
    \&
    \&
    G_2
    \arrow[hook']{ul}{}
    \\
    \&
    H
    \arrow[hook]{ul}{}
    \arrow[hook']{ur}{}
  \end{tikzcd}      
\end{equation}
such that $G_1/H$ and $G/G_2$ have the same dimension than $M$, and a generalized Cartan connection $A:TP\to\g_1$ on a $H$-principal bundle $\pi:P\to M$, construct a generalized Cartan connection $A_{\left[G_2\right]}:TP\left[G_2\right]\to\g$ in a canonical way.

In order to properly address this problem, let us establish the following auxiliary result.

\begin{lemma}
  Let $K\subset G_1\subset G$ be a chain of Lie groups, and consider a $K$-principal bundle $\pi:P\to M$. Then
  \[
    P\left[G\right]\simeq\left(P\left[G_1\right]\right)\left[G\right].
  \]
\end{lemma}
\begin{proof}
  The identification is given by the map
  \[
    \left[u,g\right]\in P\left[G\right]\mapsto\Phi\left(\left[u,g\right]\right):=\left[\left[u,e\right],g\right].
  \]
  In fact, for every $\left[\left[u,h\right],g'\right]\in\left(P\left[G_1\right]\right)\left[G\right]$, we have that
  \[
    \left[\left[u,h\right],g'\right]=\left[\left[u,e\right],hg'\right]=\Phi\left(\left[u,hg'\right]\right),
  \]
  showing that $\Phi$ is surjective. On the other hand, if $\left[u,g_1\right],\left[u,g_2\right]\in P\left[G\right]$ are such that
  \[
    \Phi\left(\left[u_1,g_1\right]\right)=\Phi\left(\left[u_2,g_2\right]\right),
  \]
  we will have that
  \[
    \left[\left[u_1,e\right],g_1\right]=\left[\left[u_2,e\right],g_2\right],
  \]
  meaning that
  \[
    \left[u_1,h^{-1}\right]=\left[u_2,e\right],\qquad hg_1=g_2
  \]
  for some $h\in G_1$. Therefore
  \[
    u_1k^{-1}=u_2,\qquad kh^{-1}=e,\qquad hg_1=g_2
  \]
  and so
  \[
    \left[\left[u_2,e\right],g_2\right]=\left[\left[u_1k^{-1},e\right],kg_1\right]=\left[\left[u_1,e\right],g_1\right].
  \]
  This shows that $\Phi$ is also injective.
\end{proof}

This lemma allows us to lift Diagram \ref{eq:GroupsDiagram} to principal bundles level:
\begin{equation}\label{eq:PrincipalBundleDiagram}
  \begin{tikzcd}[ampersand replacement=\&,row sep=.8cm,column sep=.8cm]
    \&
    P\left[G\right]
    \&
    \\
    P\left[G_1\right]
    \arrow[hook]{ur}{\gamma_{G_1}}
    \&
    \&
    P\left[G_2\right]
    \arrow[hook',swap]{ul}{\gamma_{G_2}}
    \\
    \&
    P
    \arrow[hook]{ul}{\gamma_H^1}
    \arrow[hook',swap]{ur}{\gamma_H^2}
  \end{tikzcd}      
\end{equation}

\begin{remark}
  It is interesting to note that, in view of Lemma \ref{lem:BundleIsomorphism}, the solution to this problem described below, will apply to any diagram of principal bundles
  \[
    \begin{tikzcd}[ampersand replacement=\&,row sep=.8cm,column sep=.8cm]
      \&
      Q
      \&
      \\
      P_1
      \arrow[hook]{ur}{j_1}
      \&
      \&
      P_2
      \arrow[hook',swap]{ul}{j_2}
      \\
      \&
      P
      \arrow[hook]{ul}{i_1}
      \arrow[hook',swap]{ur}{i_2}
    \end{tikzcd}      
  \]
  where $P,P_1,P_2,Q$ are $H,G_1,G_2,G$-principal bundles respectively, and the arrows indicate principal bundle immersions.
\end{remark}

The idea to extend the $\g_1/\hf$-Cartan connection on $P$ is to use Proposition \ref{prop:CartanVsPrincipalConn}; with its help, we can find a principal connection on the bundle $P\left[G_1\right]$, and then lift it through the map $\gamma_{G_1}:P\left[G_1\right]\hookrightarrow P\left[G\right]$. Afterwards, we can induce a $\g/\g_2$-Cartan connection on $P\left[G_2\right]$ using Proposition \ref{prop:CartanVsPrincipalConn} together with the map $\gamma_{G_2}:P\left[G_2\right]\hookrightarrow P\left[G\right]$.

\begin{proposition}\label{prop:extens-gener-cart}
  Given any generalized $\g_1/\hf$-Cartan connection $A:TP\to\g_1$, there exists a unique generalized $\g/\g_2$-Cartan connection $B:TP\left[G_2\right]\to\g$ with the following property: If $\widetilde{B}:TP\left[G\right]\to\g$ is the principal connection on $P\left[G\right]$ associated to $B$ and $\widetilde{A}:TP\left[G_1\right]\to\g_1$ is the corresponding principal connection for $A$, we have that
  \[
    \gamma_{G_1}^*\widetilde{B}=\widetilde{A}.
  \]
\end{proposition}

In terms of bundle maps, this correspondence proceeds as follows: Given a Cartan connection
\[
  \Gamma:P\to J^1\pi_{\left[G_1\right]},
\]
we can define a principal connection $\widetilde{\Gamma}:P\left[G\right]\to J^1\pi_{\left[G\right]}$ through
\[
  \widetilde{\Gamma}\left(\left[u,g\right]\right):v\in T_xM\mapsto T_{\gamma_{G_1}\left(\left[u,e\right]\right)}R_g\left(T_{u}\gamma_{G_1}\left(\Gamma\left(u\right)\left(v\right)\right)\right).
\]
Then the induced $\g/\g_2$-Cartan connection is the map
\begin{equation}\label{eq:InducedCartanMap}
  \Gamma^\sharp\left(\left[u,g\right]\right):v\in T_xM\mapsto T_{\gamma_{G_1}\left(\left[u,e\right]\right)}R_g\left(T_{u}\gamma_{G_1}\left(\Gamma\left(u\right)\left(v\right)\right)\right)
\end{equation}
for every $\left[u,g\right]\in P\left[G_2\right]\equiv\gamma_{G_2}\left(P\left[G_2\right]\right)$.

\subsection{Reducible Cartan connection}
\label{sec:induc-cart-conn}

What about the converse of this result? Namely, given a $\g/\g_2$-Cartan connection, will it restrict to a $\g_1/\hf$-Cartan connection on $P$? The problem with this question is that, because of the way in which the extension is defined, we have that $\widetilde{B}$ should be $\g_1$-valued, at least when it is restricted to $P\left[G_1\right]$. It poses some restrictions to the desired converse result.

\begin{proposition}\label{prop:induc-cart-conn}
  Let $\Gamma^\sharp:P\left[G_2\right]\to J^1\pi_{\left[G\right]}$ be a generalized $\g/\g_2$-Cartan connection on $P\left[G_2\right]$. If the induced principal connection
  \[
    \widetilde{\Gamma}:P\left[G\right]\to J^1\pi_{\left[G\right]}
  \]
  is such that
  \[
    \widetilde{\Gamma}\left(\gamma_{G_1}\left(\left[u,g_1\right]\right)\right)\in j^1\gamma_{G_1}\left(J^1\pi_{\left[G_1\right]}\right)
  \]
  for all $\left[u,g_1\right]\in P\left[G_1\right]$, then it reduces to a generalized $\g_1/\hf$-Cartan connection on $P$
  \[
    \Gamma:P\to J^1\pi_{\left[G_1\right]}
  \]
  through the formula
  \[
    \widetilde{\Gamma}\left(\gamma_{G_1}\left(\left[u,e\right]\right)\right)= j^1\gamma_{G_1}\left(\Gamma\left(u\right)\right).
  \]
\end{proposition}
A commutative diagram could be clearer in illustrating these matters
\begin{equation}\label{eq:ConnectionExtensionDiagram}
  \begin{tikzcd}[ampersand replacement=\&,row sep=.8cm,column sep=.8cm]
    J^1\pi_{\left[G\right]}
    \&
    \&
    P\left[G\right]
    \arrow[swap]{ll}{\widetilde{\Gamma}}
    \&
    \\
    \&
    P\left[G_1\right]
    \arrow[hook]{ur}{\gamma_{G_1}}
    \&
    \&
    P\left[G_2\right]
    \arrow[hook']{ul}{\gamma_{G_2}}
    \arrow[bend right=70,swap]{ulll}{\Gamma^\sharp}
    \\
    J^1\pi_{\left[G_1\right]}
    \arrow[]{uu}{j^1\gamma_{G_1}}
    \arrow[swap]{ur}{\left(\pi_{\left[G_1\right]}\right)_{10}}
    \&
    \&
    P
    \arrow[hook,swap]{ul}{\gamma_H^1}
    \arrow[hook',swap]{ur}{\gamma_H^2}
    \arrow[]{ll}{\Gamma}
    \arrow[hook]{uu}{\gamma_G}
  \end{tikzcd}
\end{equation}

\subsection{Examples of extensions for (generalized) Cartan connections}
\label{sec:exampl-exten-cartan}

Here we will consider some examples for the constructions devised above; the example concerning $K$-structures will be relevant when we will consider the relationship between Wise formulation of the correspondence Chern-Simons and gravity and the formulation of the correspondence in terms of bundle with the group $GL\left(m\right)$ as structure group (see Section \ref{sec:chern-simons-vari-as-Wise-problem} below).

\subsubsection{Extension of principal connections}
\label{sec:extens-princ-conn}

Let $G$ be a Lie group and $H\subset G$ a closed Lie subgroup. As a first example, let us consider the pair of diagrams
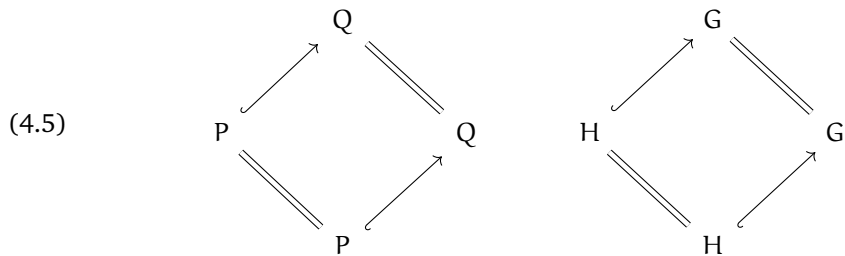
\begin{equation}\label{eq:GroupsDiagramConnec}
  \begin{tikzcd}[ampersand replacement=\&,row sep=.9cm,column sep=1cm]
    \&
    Q
    \&
    \&
    \&
    G
    \&
    \\
    P
    \arrow[hook]{ur}{}
    \&
    \&
    Q
    \arrow[equal]{ul}{}
    \&
    H
    \arrow[hook]{ur}{}
    \&
    \&
    G
    \arrow[equal]{ul}{}
    \\
    \&
    P
    \arrow[equal]{ul}{}
    \arrow[hook']{ur}{}
    \&
    \&
    \&
    H
    \arrow[equal]{ul}{}
    \arrow[hook']{ur}{}
  \end{tikzcd}      
\end{equation}
where $\pi:P\to M$ is a $H$-principal bundle and $p:Q\to M$ is a $G$-principal bundle; the left diagram is the principal bundles diagram, and the right diagram corresponds to the underlying Lie groups. In this case, Proposition \ref{prop:extens-gener-cart} reduces to the usual result on extensions of principal connections (see Proposition $6.1$ in \cite{KN1}); on the other hand, hypothesis in Proposition \ref{prop:induc-cart-conn} is equivalent to the reducibility of the principal connection on $Q$.

\subsubsection{Induced principal connection for a Cartan connection}
\label{sec:induc-princ-conn}

Let $H,G,P,Q$ be as in the previous section. We can put Proposition \ref{prop:CartanVsPrincipalConn} also in this context; to this end, let us consider the diagrams
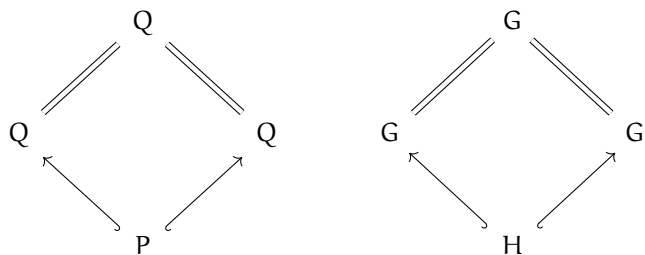
\[
  \begin{tikzcd}[ampersand replacement=\&,row sep=.9cm,column sep=1cm]
    \&
    Q
    \&
    \&
    \&
    G
    \&
    \\
    Q
    \arrow[equal]{ur}{}
    \&
    \&
    Q
    \arrow[equal]{ul}{}
    \&
    G
    \arrow[equal]{ur}{}
    \&
    \&
    G
    \arrow[equal]{ul}{}
    \\
    \&
    P
    \arrow[hook]{ul}{}
    \arrow[hook']{ur}{}
    \&
    \&
    \&
    H
    \arrow[hook]{ul}{}
    \arrow[hook']{ur}{}
  \end{tikzcd}
\]
comprising Lie groups and their principal bundle counterparts. In this setting, the additional hypothesis in Proposition \ref{prop:induc-cart-conn} is automatically fulfilled. So, once we realize that $Q\simeq P\left[G\right]$ via Lemma \ref{lem:BundleIsomorphism}, Proposition \ref{prop:extens-gener-cart} and Proposition \ref{prop:induc-cart-conn} are nothing but Proposition \ref{prop:CartanVsPrincipalConn}.

\subsubsection{$K$-structures on space-time}
\label{sec:k-structures-space}

The previous scheme can also be applied to the case in which the original relationship between Palatini gravity (formulated on a $SO\left(2,1\right)$-subbundle of $LM$) and Chern-Simons field theory (on a $SO\left(2,1\right)\ltimes\mR^3$-principal bundle) is promoted to a relationship between gravity formulated in the bundle $LM$ and Chern-Simons theory on a bundle with structure group $GL\left(3\right)\ltimes\mR^3$. Diagram \eqref{eq:GroupsDiagram} then becomes
\begin{equation}\label{eq:GroupsDiagramSpec}
  \begin{tikzcd}[ampersand replacement=\&,row sep=.9cm,column sep=.6cm]
    \&
    GL\left(m\right)\ltimes\mR^m
    \&
    \\
    SO\left(p,q\right)\ltimes\mR^m
    \arrow[hook]{ur}{}
    \&
    \&
    GL\left(m\right)
    \arrow[hook']{ul}{}
    \\
    \&
    SO\left(p,q\right)
    \arrow[hook]{ul}{}
    \arrow[hook']{ur}{}
  \end{tikzcd}      
\end{equation}
that in terms of principal bundles turns out to be
\begin{equation*}
  \begin{tikzcd}[ampersand replacement=\&,row sep=.9cm,column sep=.9cm]
    \&
    AM
    \&
    \\
    O_\zeta^{\text{aff}}
    \arrow[hook]{ur}{}
    \&
    \&
    LM
    \arrow[hook']{ul}{}
    \\
    \&
    O_\zeta
    \arrow[hook]{ul}{}
    \arrow[hook']{ur}{}
  \end{tikzcd}      
\end{equation*}
Here $O_\zeta\subset LM$ and $O_\zeta^{\text{aff}}\subset AM$ are the $SO(p,q)$- and $SO(p,q)\ltimes\mR^m$-structures respectively, defined through
\[
  O_\zeta:=\bigcup_{x\in M}\left\{u:\mR^m\to T_xM:\zeta(u\left(v\right),u\left(w\right))=\eta\left(v,w\right)\text{ for all }v,w\in\mR^m\right\}
\]
and
\[
  O_\zeta^{\text{aff}}:=\bigcup_{x\in M}\left\{u:\mR^m\to A_xM:\zeta(\beta\left(u\left(v\right)\right),\beta\left(u\left(w\right)\right))=\eta\left(v,w\right)\text{ for all }v,w\in\mR^m\right\}
\]
for some metric $\zeta:M\to\Sigma:=LM/K$ with $\left(p,q\right)$-signature.

Let us use the local description for the frame bundle and its affine counterpart (see Appendix \ref{sec:lift-conn-affine}) in order to show how the proposed extension works in this case. Let us consider $\left(x^\mu,f^\nu_i\right)$ the natural coordinates of an element in $O_\zeta$; it means that
\[
  g_{\mu\nu}f^\mu_if^\nu_j=\eta_{ij},
\]
where $\zeta=g_{\mu\nu}dx^\mu\otimes dx^\nu$. The map
\[
  \phi:LM\times_{GL\left(m\right)}A\left(m\right)\to AM
\]
constructed in Proposition \ref{prop:FrameBundleExtension}, Appendix \ref{sec:affine-frame-bundle}, reads in these coordinates
\begin{equation}\label{eq:LocalIdentificationAffineFrames}
  \phi\left(\left[\left(x^\mu,f^\nu_i\right),\left(a^i_j,0\right)\right]\right)=\left(x^\mu,a^j_if_j^\nu\right).
\end{equation}
Now, from a local data
\begin{equation}\label{eq:CartanConnectionLocalData}
  \left(A_{O_\zeta}\right)_U=\left(\Gamma^\alpha_{\gamma\beta}dx^\beta\otimes E_\alpha^\gamma,\sigma^\alpha_\beta dx^\beta\otimes e_\alpha\right)
\end{equation}
for the Cartan connection $A_{O_\zeta}\in\Omega^1\left(O_\zeta,\kf\ltimes\mR^m\right)$, using Proposition \ref{prop:GeneralizedCartanConnection} we obtain a map
\[
  \Gamma_{A_{O_\zeta}}:O_\zeta\to J^1\left(\left.\left(\beta\circ\tau\right)\right|_{O_\zeta^{\text{aff}}}\right)
\]
that is locally given by
\[
  \Gamma_{A_{O_\zeta}}\left(x^\mu,f^\nu_i\right)=dx^\mu\otimes\left(\frac{\partial}{\partial x^\mu}-f^\gamma_i\Gamma_{\gamma\mu}^\alpha\frac{\partial}{\partial e^\alpha_i}-\sigma_\mu^\alpha\frac{\partial}{\partial v^\alpha}\right)
\]
\begin{remark}\label{rem:MetricityMeaning}
As an aside comment that will become important later, it can be proved that relation \eqref{eq:MetricityConditionLocal} below implies that
\[
  F_\mu:=\frac{\partial}{\partial x^\mu}-f^\gamma_i\Gamma_{\gamma\mu}^\alpha\frac{\partial}{\partial e^\alpha_i}
\]
is a vector field tangent to $O_\zeta$, and so
\[
  \Gamma_{A_{O_\zeta}}\left(x^\mu,f^\nu_i\right)\in J^1\left(\left.\left(\beta\circ\tau\right)\right|_{O_\zeta^{\text{aff}}}\right).
\]
Namely, as expected, the metricity condition implies that $\Gamma_{A_{O_\zeta}}$ takes values in the jet bundle of the affine subbundle determined by the metric $\zeta$.
\end{remark}
Now, using Equation \eqref{eq:InducedCartanMap}, we obtain the induced Cartan connection given by the map
\[
  \Gamma_{A_{O_\zeta}}^\sharp\left(\left[\left(x^\mu,f^\nu_i\right),\left(a^i_j,0\right)\right]\right)=dx^\mu\otimes\left(\frac{\partial}{\partial x^\mu}-a_i^jf^\gamma_j\Gamma_{\gamma\mu}^\alpha\frac{\partial}{\partial e^\alpha_i}-\sigma_\mu^\alpha\frac{\partial}{\partial v^\alpha}\right),
\]
and using identification \eqref{eq:LocalIdentificationAffineFrames}, it will become
\[
  \Gamma_{A_{O_\zeta}}^\sharp\left(x^\mu,e^\nu_i\right)=dx^\mu\otimes\left(\frac{\partial}{\partial x^\mu}-e^\gamma_j\Gamma_{\gamma\mu}^\alpha\frac{\partial}{\partial e^\alpha_i}-\sigma_\mu^\alpha\frac{\partial}{\partial v^\alpha}\right).
\]
This is nothing but the Cartan connection associated to local data \eqref{eq:CartanConnectionLocalData}, when it is considered as providing a $\mathfrak{a}\left(m\right)/\mathfrak{gl}\left(m\right)$-Cartan connection on the frame bundle $LM$.

\section{Variational problems for Chern-Simons field theory and gravity}
\label{sec:vari-probl-chern}

We will use the present section to introduce the variational problem for Chern-Simons field theory used by Wise in \cite{wise2009symmetric,Wise_2010}, and the variational problem for gravity with basis, both in a form suitable for the purposes of this article. Concretely, we will try to find a formulation for these variational problems fitting in the scheme devised by Gotay in the pioneering works \cite{GotayCartan,Gotay1991203}. It means that we need to find for each of these descriptions a triple
\[
  \left(\pi:P\to M,\lambda,\cI\right),
\]
where $\pi:P\to M$ is a bundle on the base space $M$, $\lambda\in\Omega^n\left(P\right)$ is an $n$-form (where $n=\dim{M}$), and $\cI\subset\Omega^\bullet\left(P\right)$ is a differential ideal in the exterior algebra of $P$ (a so called \emph{exterior differential system}, see \cite{BryantNine,BCG,CartanBeginners}). As Gotay explained in the article referenced above, with these data it is possible to formulate a variational problem in the following way: To find the stationary sections $\sigma:U\subset M\to P$ for the action
\[
  S\left[\sigma\right]:=\int_U\sigma^*\lambda
\]
subject to the constraints $\sigma^*\alpha=0$ for all $\alpha\in\cI$. In the present article we will use the term \emph{Gotay variational problem} to refer to this kind of variational problems; in this secrtion we will describe Chern-Simons field theory and gravity from this viewpoint.

\subsection{Wise variational problem}
\label{sec:wise-vari-probl}

Our first objective is to find a variational problem of this type for the Chern-Simons field theory, as described in the articles by Wise \cite{wise2009symmetric,Wise_2010}. To do this, we will use as a starting point the description of this field theory given by Tejero Prieto \cite{TejeroPrieto2004}, which uses a formulation in terms of local variational problems. First, we will study how the canonical forms of the jet space of a principal fiber bundle can be used to globalize this collection of local variational problems. Then, we will introduce a constraint that will relate the degrees of freedom of the underlying principal bundle with the degrees of freedom of the Cartan connection. It should be clarified that this constraint is absent in Wise's description because the relevant degrees of freedom are associated only with the connection, whereas in our approach, by using the frame bundle as the underlying principal bundle, we have degrees of freedom that can be used to represent the coframe (see Remark \ref{rem:Kk0Constraint} below).

\subsubsection{Local formulation for Wise variational problem}
\label{sec:wise-vari-probl-1}
 
Let us suppose that we have a Cartan geometry modelled on the pair $\left(G,H\right)$. The Cartan connection $A$ associated to this geometry can be described locally by a collection of pairs $\left(U,A_U\right)$, where $U\subset M$ are open sets and $A_U\in\Omega^1\left(U,\g\right)$ are $\g$-valued $1$-forms such that the map
\[
  \pi_\hf\circ\left.A_U\right|_x:T_xU\to\g/\hf
\]
is a linear isomorphism. These forms are related to the Cartan connection $A$ through a section $s_U:U\to P$, via the map
\[
  A_U=s_U^*A.
\]
Now, according to Freed \cite{Freed:1992vw}, on every trivializing neighborhood $U\subset M$ it is possible to define an action for Chern-Simons through
\[
  S_U\left[s_U,A\right]:=\int_Us_U^*Tq\left(A,F\right).
\]
Modulo some topological assumptions regarding the structure group of the principal group, it can be proved that this action is independent of the section $s_U$ involved in its definition. Moreover, whenever $\dim{M}=3$ and the Lie group $H$ is simply connected, any $H$-principal bundle on $M$ is necessarily trivial, so that this prescription gives rise to a well-defined variational problem on the whole principal bundle.

On the other hand, Tejero Prieto \cite{TejeroPrieto2004} is able to give global sense to this collection of local actions by requiring that, whenever the domains intersect, the associated Euler-Lagrange equations are the same for any of the local actions. Thus, we will interpret the Wise variational problem as a collection of local actions
\[
  S_U\left[A\right]:=\int_UTq\left(A_U,F_U\right)
\]
given by the local description $\left(U,A_U\right)$ of a Cartan connection, because its Euler-Lagrange equations will coincide on the intersection of the corresponding domains.

Using Proposition~\ref{prop:CartanVsPrincipalConn}, we consider the equivalent variational problem on the bundle $\pi_\gamma:U_\gamma\to M$; therefore, we can see the Wise variational problem as a problem whose fields are principal connections $\widetilde{A}$ on $P\left[G\right]$, and the action is given by the formula 
\[
  S_U\left[\widetilde{A}\right]:=\int_UTq\left(\widetilde{A}_U,\widetilde{F}_U\right),
\]
where $\widetilde{A}_U=\widetilde{s}_U^*\widetilde{A}$ is the local description of the principal connection $\widetilde{A}$, with
\[
  \widetilde{s}_U\left(x\right)=\gamma\left(s_U\left(x\right)\right)=\left[s_U\left(x\right),e\right].
\]

\subsubsection{Wise variational problem in jet bundle formulation}
\label{sec:wise-vari-probl-2}

It is our next aim to formulate Wise variational problem in terms of a Lagrangian form and a bundle; we are looking for a formulation in which the integrand $Tq\left(A_U,F_U\right)$ in the action comes from the Lagrangian form through pullback along a section of a suitable bundle, and this section is uniquely determined by the Cartan connection. Recalling the discussion carried out in Section \ref{sec:bundle-cart-conn}, we will choose the bundle
\[
  \pi\circ\text{pr}_1^\gamma:\gamma^*\left(J^1\pi_{\left[G\right]}\right)\to M
\]
as the bundle whose sections can be put in one-to-one correspondence with $\g/\hf$-Cartan connections on $P$.

Our next task is to find the Lagrangian form. To this end, consider a generalized Cartan connection $\Gamma:P\to J^1\pi_{\left[G\right]}$, with local description $\left\{A_U\right\}$. Then Remark \ref{rem:CaartanConnectionAndCanonicalConnection} tells us that
\[
  A_U=\left(\Gamma\circ s_U\right)^*\theta_{J^1\pi_{\left[G\right]}},
\]
where $\theta_{J^1\pi_{\left[G\right]}}\in\Omega^1\left(J^1\pi_{\left[G\right]},\g\right)$ is the canonical connection form and $s_U:U\to P$ is a local section for the principal bundle $\pi:P\to M$. Let us define
\begin{equation}\label{eq:PullbackCanonicalConnection}
  \theta_{J^1\pi_{\left[G\right]}}^*:=\left(\text{pr}_2^\gamma\right)^*\theta_{J^1\pi_{\left[G\right]}}
\end{equation}
and
\[
  \Theta_{J^1\pi_{\left[G\right]}}^*:=\left(\text{pr}_2^\gamma\right)^*\Theta_{J^1\pi_{\left[G\right]}}
\]
where $\text{pr}_2^\gamma:\gamma^*\left(J^1\pi_{\left[G\right]}\right)\to J^1\pi_{\left[G\right]}$ is the horizontal projection in the pullback diagram in Corollary \ref{cor:GeneralizedCartan}. Therefore, the Lagrangian form is given by
\[
  \lambda_{CS}:=Tq\left(\theta_{J^1\pi_{\left[G\right]}}^*,\Theta_{J^1\pi_{\left[G\right]}}^*\right)\in\Omega^{2k-1}\left(\gamma^*\left(J^1\pi_{\left[G\right]}\right)\right);
\]
we have that
\[
  L_U:=Tq\left(A_U,F_U\right)=\left(s_U,\Gamma\circ s_U\right)^*\lambda_{CS}.
\]
Thus, we obtain the following corollary.
\begin{corollary}\label{cor:wise-vari-probl}
  The extremals of the local variational problem determined by the Lagrangian $L_U$ are in one to one correspondence with the local extremals of the variational problem given by the triple
  \[
    \left(\pi\circ\text{pr}_1^\gamma:\gamma^*\left(J^1\pi_{\left[G\right]}\right)\to M,\lambda_{CS},0\right).
  \]
\end{corollary}
\begin{proof}
  The correspondence is given by
  \[
    \left(A_U,s_U\right)\mapsto\left(s_U,\Gamma\circ s_U\right)\in\left(\pi\circ\text{pr}_1^\gamma\right)^{-1}\left(U\right)\subset\gamma^*\left(J^1\pi_{\left[G\right]}\right).\qedhere
  \]
\end{proof}

\subsubsection{Wise variational problem for first order geometries}
\label{sec:wise-vari-probl-3}

Recall \cite{sharpe1997differential,Barakat2004} that a pair $\left(P,A\right)$, where $\pi:P\to M$ is an $H$-principal bundle and $A$ is a $\g/\hf$-Cartan connection, is called a \emph{first order geometry} if and only if the representation
\[
  \text{Ad}:H\to\mathop{\text{GL}}{\left(\g/\hf\right)}
\]
is faithful. Then we have the following result.
\begin{proposition}
  $P$ admits a first order geometry if and only if it is an $H$-structure.
\end{proposition}
\begin{proof}
  Recall that a Cartan connection $A$ on $P$ induces a family of isomorphisms
  \[
    \phi_u:T_xM\to\g/\hf,\qquad u\in\pi^{-1}\left(x\right)
  \]
  such that $\phi_{uh}=\text{Ad}_{h^{-1}}\circ\phi_u$. Thus, the isomorphism is given by
  \[
    u\in P\mapsto\left(\phi_u^{-1}\left(w_1\right),\cdots,\phi_u^{-1}\left(w_m\right)\right),
  \]
  where $\left\{w_1,\cdots,w_m\right\}$ is a basis for $\g/\hf$.
\end{proof}
In the present section we will assume that $P$ admits a first order geometry; according to the previous proposition, it means that $P$ can be considered a subbundle of the bundle of frames $LM$. Therefore, for every element $v_x\in T_xM$ and $u\in\pi^{-1}\left(x\right)$, we have two ways to represent it, namely
\[
  v_x=\phi_u^{-1}\left(w_u\right)=u\left(c_u\right)
\]
for some elements $w_u\in\g/\hf$ and $c_u\in\mR^m$; it gives rise to a linear isomorphism
\begin{equation}\label{eq:CartanConnectionConstraint}
  \kappa_u:\mR^m\to\g/\hf:c\mapsto\phi_u\circ u\left(c\right).
\end{equation}

\begin{example}[Affine connections]\label{example:AffineConnections}
  Let us calculate this isomorphism for $\mathfrak{a}\left(m\right)/\mathfrak{gl}\left(m\right)$-Cartan connections on $LM$, the so called \emph{generalized affine connections} \cite{KN1}. Assuming that the local version of a Cartan connection on $LM$ is
  \[
    \widetilde{\omega}_U=\left(\Gamma^\alpha_{\gamma\beta}dx^\beta\otimes E_\alpha^\gamma,\sigma^\alpha_\beta dx^\beta\otimes e_\alpha\right)
  \]
  for some local functions $\Gamma^\alpha_{\beta\gamma}$ and $\sigma^\alpha_\beta$, from Equation \eqref{eq:AffineConnectionGlobal} we obtain
  \[
    \phi_{\left(x^\alpha,e^\beta_i\right)}=e_\alpha^i\sigma^\alpha_\beta dx^\beta\otimes e_i.
  \]
  Now, recall that in this case $\mathfrak{a}\left(m\right)/\mathfrak{gl}\left(m\right)=\mR^m$ in canonical fashion, so that $\kappa_{\left(x^\alpha,e^\beta_i\right)}$ is a linear endomorphism of $\mR^m$. Therefore, the matrix for this morphism in terms of the canonical basis $\left\{e_i\right\}$ will become
  \[
    \left[\kappa_{\left(x^\alpha,e^\beta_i\right)}\right]_j^i=e_\alpha^ie^\beta_j\sigma^\alpha_\beta.
  \]
  According to the classical definition \cite{KN1}, a Cartan connection on $LM$ is an \emph{affine connection} if and only if $\kappa_{\left(x^\alpha,e^\beta_i\right)}$ is the identity on $\mR^m$, or equivalently
  \[
    e_\alpha^ie^\beta_j\sigma^\alpha_\beta=\delta_j^i.
  \]
  We can see this constraint in terms of the jet bundle description for the Cartan connection (see Section \ref{sec:canon-gener-cart}). In order to proceed, let us define the map
  \begin{equation}\label{eq:ImmersionDefinition}
    j:J^1\tau\hookrightarrow J^1\left(\tau\circ\beta\right):j_x^1s\mapsto j^1_x\left(\gamma\circ s\right)-\left.\varphi\right|_{s\left(x\right)},
  \end{equation}
  where $\varphi$ indicates the canonical $\mR^m$-valued $1$-form on $LM$, which can be seen as an element of
  \[
    T_x^*M\otimes\mR^m\subset T^*_xM\otimes\mathfrak{a}\left(m\right)\simeq T_x^*M\otimes V_{\gamma\left(s\left(x\right)\right)}\left(AM\right)
  \]
  and thus acts on the affine space $J^1_{\gamma\left(s\left(x\right)\right)}\left(\tau\circ\beta\right)$. In local terms, we have that
  \[
    \left.\varphi\right|_{\left(x^\alpha,e^\beta_i\right)}=e^i_\beta dx^\beta\otimes e_i
  \]
  and if an element $j_x^1s\in J^1\tau$ has coordinates $\left(x^\alpha,e_i^\beta,e^\beta_{i\gamma}\right)$, then it represents the map
  \[
    j_x^1s:\frac{\partial}{\partial x^\alpha}\mapsto\frac{\partial}{\partial x^\alpha}+e^\beta_{i\alpha}\frac{\partial}{\partial e^\beta_i};
  \]
  additionally, the map $\gamma:LM\to AM$ simply reads
  \[
    \gamma\left(x^\alpha,e_i^\beta\right)=\left(x^\alpha,e_i^\beta,0\right).
  \]
  Then, from Definition \eqref{eq:ImmersionDefinition} we can conclude that
  \begin{equation}\label{eq:jInLocalTerms}
    j\left(x^\alpha,e_i^\beta,e^\alpha_{i\gamma}\right)=\left(x^\alpha,e_i^\beta,0,e^\alpha_{i\gamma},-e^\alpha_ie^i_\beta\right),
  \end{equation}
  as required.
\end{example}

Our next task is to adapt the Equation \eqref{eq:CartanConnectionConstraint} to the description discussed in Section \ref{sec:cart-conn-as}, where Cartan connections were considered as sections in a bundle. The main reason to do that is that, just as we saw in the previous example, the pullback bundle $\gamma^*\left(J^1\pi_{\left[G\right]}\right)$ would contain degrees of freedom associated to the principal bundle $P$ and the $\g/\hf$-part of the Cartan connection. In fact, given an element $\overline{u}=\left(u,j_x^1s\right)\in\gamma^*\left(J^1\pi_{\left[G\right]}\right)$, we can construct the map
\[
  {\phi}_{\overline{u}}:T_xM\to\g/\hf
\]
in the following way: Given a tangent vector $v_x\in T_xM$, fix a lift $\widehat{v}_u\in T_uP$, and use it to construct the tangent vector
\[
  W_{s\left(x\right)}:=T_u\gamma\left(\widehat{v}_{u}\right)-T_xs\left({v}_{x}\right)\in T_{s\left(x\right)}\left(P\left[G\right]\right).
\]
Then we have that
\[
  T_{s\left(x\right)}\pi_{\left[G\right]}\left(T_u\gamma\left(\widehat{v}_{u}\right)\right)=v_x
\]
and so $W_{s\left(x\right)}\in V_{s\left(x\right)}\left(P\left[G\right]\right)$; it means in particular that there exists an element $\widetilde{W}\left(\widehat{v}_u\right)\in\g$ such that
\[
  W_{s\left(x\right)}=\left(\widetilde{W}\left(\widehat{v}_u\right)\right)_{P\left[G\right]}\left(s\left(x\right)\right).
\]
Now, changing the lift $\widehat{v}_u$ in this definition will produce a shift in $\widetilde{W}\left(\widehat{v}_u\right)$ by an element living in $\hf$; therefore, we can define the map ${\phi}_{\overline{u}}:T_xM\to\g/\hf$ by projecting into the quotient by $\hf$, namely
\[
  {\phi}_{\overline{u}}\left(v_x\right):=\left[\widetilde{W}\left(\widehat{v}_u\right)\right]_\hf,
\]
where $\left[\cdot\right]_\hf$ indicates the equivalence class in $\g/\hf$.

Using Example \ref{example:AffineConnections} and Equation \eqref{eq:CartanConnectionConstraint}, we can define a map
\[
  \kappa_{\overline{u}}:=\phi_{\overline{u}}\circ u:\mR^m\to\g/\hf
\]
for every $\overline{u}=\left(u,j_x^1s\right)\in\gamma^*\left(J^1\pi_{\left[G\right]}\right)$. It will allow us to identify the degrees of freedom mentioned above by means of the constraint
\begin{equation}\label{eq:FirstOrderGeometryConstraint}
  \kappa_u=\kappa_0
\end{equation}
for some fixed isomorphism $\kappa_0:\mR^m\to\g/\hf$, under the assumption that $P$ admits a first order geometry. So we are ready to introduce the following definition.
\begin{definition}[Wise variational problem for first order geometries]\label{def:WiseVariationalProblem}
  Let $\pi:P\to M$ be an $H$-principal bundle admitting a first order geometry associated to the pair $\left(G,H\right)$; fix an isomorphism $\kappa_0:\mR^m\to\g/\hf$. The \emph{Wise variational problem} is the triple
  \[
    \left(\pi\circ\text{pr}_1^\gamma:\gamma^*\left(J^1\pi_{\left[G\right]}\right)\to M,\lambda_{CS},\mathcal{K}_{\kappa_0}\right).
  \]
  where $\mathcal{K}_{\kappa_0}$ is the EDS induced by the constraint \eqref{eq:FirstOrderGeometryConstraint}.
\end{definition}

\begin{remark}[On the the constraint $\mathcal{K}_{\kappa_0}$]\label{rem:Kk0Constraint}
  The introduction of the constraint $\mathcal{K}_{\kappa_0}$ implies a departure from the scheme devised by Wise in the previously cited works, where the degrees of freedom associated to the $\g/\hf$-part of the Cartan connection are the ones devoted to describe the vielbein when the relationship with gravity is established. In our present description, we are choosing to use the degrees of freedom of the underlying principal bundle in order to describe it, and the r\^{o}le of the constraint $\mathcal{K}_{\kappa_0}$ is to enforce the identification of the $\g/\hf$-part of the Cartan connection with the group coordinates in the principal bundle $P$. For example, when $P=LM,H=GL\left(m\right),G=A\left(m\right)$, we can deal with the constraint $\mathcal{K}_{\kappa_0}$ in a straightforward manner. In fact, because of the Equation \eqref{eq:ImmersionDefinition}, we have that a section $\sigma:U\subset M\to\gamma^*\left(J^1\left(\tau\circ\beta\right)\right)$ is integral for this EDS (with $\kappa_0=\text{id}$, given the identification $\g/\hf=\mR^m$ we have at our disposal is this case) if and only if there exists a section $s:U\to J^1\tau$ such that
  \[
    \sigma\left(x\right)=\left(\tau_{10}\left(s\left(x\right)\right),j\left(s\left(x\right)\right)\right).
  \]
  Thus we can fulfill this constraint by using this form for the sections we consider in the variational problem.
\end{remark}

\subsection{Variational problem for gravity with basis}
\label{sec:vari-probl-grav}

It is usual \cite{MR974271} to describe gravity with a pair $\left(e^\mu_i,\omega^i_j\right)$, where $e^\mu_i$ are the components of a local basis of the tangent bundle $TM$ respect to some coordinates $x^\mu$, and $\omega^i_j$ are a set of local $1$-form, the so called \emph{spin connection forms}. This pair can be represented in terms of the bundle of frames $LM$ on the space-time: Fixing a common open domain $U\subset M$, $e_i^\mu$ gives rise to a local section $\sigma$ of $LM$ through the formula
\[
  x\in U\mapsto\sigma\left(x\right):=\left(e_1^\mu\left(x\right)\frac{\partial}{\partial x^\mu},\cdots,e_m^\mu\left(x\right)\frac{\partial}{\partial x^\mu}\right).
\]
The forms $\omega_j^i$ together with the section $\sigma$ can be used in order to define the Christoffel symbols $\Gamma^\mu_{\nu\sigma}$ of the underlying connection, through the so called \emph{vielbein postulate}
\begin{equation}\label{eq:TetradPostulate}
  \omega_j^i=e_\mu^i\left(e^\nu_j\Gamma^\mu_{\nu\sigma}+\frac{\partial e_j^\mu}{\partial x^\sigma}\right)dx^\sigma.
\end{equation}
The connection form on $LM$\footnote{Or a subbundle...} associated to $\left\{\omega_j^i\right\}$ can be retrieved by the formula \cite{Naka,KN1}
\[
  \left.\left(\omega_{LM}\right)^i_j\right|_u=e_\mu^ie^\nu_j\Gamma^\mu_{\nu\sigma}dx^\sigma+e^i_\mu de^\mu_j
\]
where $u\in LM$ is represented by the set of coordinates $\left(x^\mu,e^\mu_i\right)$. This gives rise to a locally defined action
\begin{equation}\label{eq:PalatiniLocalAction}
  S_{\text{PG}}:=\int_U\epsilon_{ijkl}\eta^{kp}e^i_\mu e^j_\nu dx^\mu\wedge dx^\nu\wedge\Omega^l_p,
\end{equation}
where $\Omega_j^i$ is the local curvature $2$-form associated to $\omega^i_j$. A constraint that should be adopted on the set of forms $\omega^i_j$ is that the underlying connection form is $\mathfrak{o}\left(3,1\right)$-valued, namely
\begin{equation}\label{eq:SkewSymmetricConnection}
  \eta^{ij}\omega^k_j+\eta^{kj}\omega_j^i=0.
\end{equation}

Before to continue, we have to deal with the constraints given by Eqs. \eqref{eq:SkewSymmetricConnection}; according to formula \eqref{eq:TetradPostulate}, it is equivalent to
\begin{equation}\label{eq:MetricityConditionLocal}
  \frac{\partial g^{\mu\nu}}{\partial x^\sigma}+g^{\mu\rho}\Gamma^\nu_{\rho\sigma}+g^{\nu\rho}\Gamma^\mu_{\rho\sigma}=0
\end{equation}
namely, the connection associated to the symbols $\Gamma^\mu_{\nu\sigma}$ is metric with respect to the metric associated to the vielbein,
\[
  g^{\mu\nu}=\eta^{ij}e^\mu_ie^\nu_j.
\]
At jet bundle level it is equivalent to the set of equations
\begin{equation}\label{eq:MetricityLocalEquation}
  \left(\eta^{kj}e_\mu^i+\eta^{ij}e_\mu^k\right)\left(-e^\mu_{j\sigma}+\frac{\partial e_j^\mu}{\partial x^\sigma}\right)=0.
\end{equation}
Thus, a solution for gravity with basis is a section
\[
  \sigma:U\subset M\to J^1\tau
\]
which is both an extremal for the functional \eqref{eq:PalatiniLocalAction} and also verifies the condition \eqref{eq:MetricityLocalEquation}.

On the other hand, there is a formulation for field theory in which the action is calculated using a Lagrangian density, that is a bundle map
\[
  \cL:J^1q\to\wedge^4\left(T^*M\right),
\]
where $q:E\to M$ is a bundle on $M$. Namely, for every section $\sigma:M\to E$ of $q$, we have a map
\[
  \cL\circ j^1\sigma:M\to\wedge^4\left(T^*M\right),
\]
which is a $4$-form on $M$, and so you can integrate,
\[
  S_{PG}\left[\sigma\right]:=\int_M\cL\circ j^1\sigma.
\]
We want to find this kind of formulation for gravity with basis; in order to carry out this task, it is necessary to identify the bundle $E$, and then to write down a Lagrangian density on this bundle. Because $e_i^\alpha$ are part of the degrees of freedom for this flavor of gravity, the bundle $E$ should include it, namely, we are searching for a bundle of the form
\[
  E=LM\times_M\star,
\]
where $\star$ stands for the degrees of freedom associated to $\omega_j^i$. As we said before, these data corresponds to the specification of a linear connection on $LM$; from a geometrical viewpoint, the relevant bundle in this regard is the connection bundle $\overline{\tau}:C\left(LM\right)\to M$, defined in such a way that it fits in the following diagram
\[
  \begin{tikzcd}[ampersand replacement=\&,row sep=1.3cm,column sep=1.6cm]
    J^1\tau
    \arrow{r}{\tau_{10}}
    \arrow[swap]{d}{p_{GL\left(m\right)}^{J^1\tau}}
    \&
    LM
    \arrow{d}{\tau}
    \\
    C\left(LM\right)
    \arrow{r}{\overline{\tau}}
    \&
    M
  \end{tikzcd}      
\]
Therefore, a section of the connection bundle is equivalent to a equivariant section of the projection $\tau_{10}$. Now, recall that an element of $J^1\tau$ is a linear map $m:T_xM\to T_u\left(LM\right)$ such that
\[
  T_u\tau\circ m=\text{id}_{T_xM}.
\]
Using the data $\omega_j^i$, this section can be constructed according to the formula
\begin{align*}
  \Gamma\left(u\right)&:=dx^\mu\otimes\left(\frac{\partial}{\partial x^\mu}\right)_u^H\\
  &=dx^\mu\otimes\left(\frac{\partial}{\partial x^\mu}-e_j^\rho\Gamma^\sigma_{\rho\mu}\frac{\partial}{\partial e^\sigma_j}\right)
\end{align*}
where $u=\left(X_1,\cdots,X_m\right)\in LM$ and the symbols $\Gamma^\mu_{\nu\sigma}$ are calculated through Equation \eqref{eq:TetradPostulate}; here $\left(\cdot\right)^H_u$ indicates the horizontal lift of a vector field to $T_uLM$. Thus the bundle describing this flavor of gravity will become
\[
  E=LM\times_M C\left(LM\right).
\]
Moreover, it can be proved \cite{MR0315624,springerlink:10.1007/PL00004852} that this bundle is isomorphic as affine bundle on $LM$ to the first jet bundle $J^1\tau$; in terms of the induced coordinates $\left(x^\mu,e_i^\nu,e^\sigma_{k\rho}\right)$, this correspondence is given by the formula
\begin{equation}\label{eq:OmegaVsJets}
  \Gamma^\mu_{\nu\sigma}=-e^k_\nu e^\mu_{k\sigma}.
\end{equation}
Then, the relevant bundle in the description of gravity with basis will be the first jet bundle of the frame bundle.

Having identified the basic bundle, we need to write down a Lagrangian density and to encode the constraint imposed by Eq. \eqref{eq:SkewSymmetricConnection}. In order to proceed, let us define the \emph{canonical connection form}
\[
  \left.\theta_{J^1\tau}\right|_{j_x^1s}:=e^i_\mu\left(de^\mu_j-e^\mu_{j\rho}dx^\rho\right)\otimes E^j_i
\]
where $j_x^1s=\left(x^\mu,e^\nu_j,e^\sigma_{k\rho}\right)$ are the induced coordinates on $J^1\tau$; it can be proved that this formula defines a $\mathfrak{gl}\left(m\right)$-valued $1$-form on $J^1\tau$, and that it becomes a connection form on the $GL\left(m\right)$-principal bundle
\[
  p_{GL\left(m\right)}^{J^1\tau}:J^1\tau\to C\left(LM\right).
\]
Let
\[
  \rho\left(x\right):=\left(x^\mu,e^\mu_i\left(x\right),e^\nu_{k\sigma}\left(x\right)\right)
\]
be a local description for a section $\rho:U\to J^1\tau$; then from the tetrad postulate \eqref{eq:TetradPostulate} and using Eq. \eqref{eq:OmegaVsJets}, we have that
\[
  \rho^*\theta_{J^1\tau}=\left(\omega_{LM}\right)^i_j\otimes E^j_i.
\]
Let $\Omega_{LM}\in\Omega^2\left(J^1\tau,\mathfrak{gl}\left(m\right)\right)$ be the curvature form associated to $\gamma$; naturatility of the pullback implies that
\[
  \rho^*\Theta_{J^1\tau}=\left(\Omega_{LM}\right)^i_j\otimes E^j_i
\]
at curvature forms level. The formula
\[
  \left.\varphi\right|_{j_x^1s}:=e^i_\mu dx^\mu\otimes e_j
\]
also defines a global $\mR^m$-valued $1$-form on $J^1\tau$; out from this form and the bilinear map $\eta:\mR^m\times\mR^m\to\mR$, we can define a $GL\left(m\right)$-valued $\left(m-2\right)$-form as follows: First, define the $\wedge^{m-2}\mR^m$-valued $\left(m-2\right)$-form given by the formula
\[
  \varphi^{m-2}:=\overbrace{\varphi\wedge\cdots\wedge\varphi}^{m-2}.
\]
Then, use the star map $\star:\wedge^{m-2}\mR^m\to\wedge^{2}\mR^m$ determined by the bilinear form $\eta$ to define a $\wedge^{2}\mR^m$-valued $\left(m-2\right)$-form $\star\left(\varphi^{m-2}\right)$; finally, use the isomorphism $\eta:\mR^m\to\left(\mR^m\right)^*$ to define a map
\[
  \eta^\sharp:\wedge^2\mR^m\to\left(\mR^m\right)^*\otimes\mR^m=\mathfrak{gl}\left(m\right)
\]
giving rise to a $\mathfrak{gl}\left(m\right)$-valued $\left(m-2\right)$-form, called \emph{Sparkling form}
\[
  \varphi_{m-2}^\sharp:=\eta^\sharp\left(\varphi^{m-2}\right).
\]
Additionally, $\eta$ induces a bilinear pairing between $\mathfrak{gl}\left(m\right)$-forms on $J^1\tau$. Using the naturality properties of the pullback, we can conclude that (in the case $m=4$)
\[
  \rho^*\left(\eta\left(\varphi_{2}^\sharp\stackrel{\wedge}{,}\Gamma\right)\right)=\epsilon_{ijkl}\eta^{kp}e^i_\mu e^j_\nu dx^\mu\wedge dx^\nu\wedge\Omega^l_p;
\]
also, for the case $m=3$ it would result
\[
  \rho^*\left(\eta\left(\varphi_{1}^\sharp\stackrel{\wedge}{,}\Gamma\right)\right)=\epsilon_{ijk}\eta^{kp}e^i_\mu dx^\mu\wedge\Omega^j_p.
\]
These considerations allows us to define the \emph{Lagrangian form for gravity with basis}, that becomes
\begin{equation}\label{eq:GravityWithBasisLagrangian}
  \lambda_{\text{PG}}:=\eta\left(\varphi_{m-2}^\sharp\stackrel{\wedge}{,}\Gamma\right).
\end{equation}
Strictly speaking, $\lambda_{PG}$ is not a $\tau_1$-horizontal form on $J^1\tau$, and so it is not associated to a Lagrangian density; if we wanted to deal with the underlying variational problem as an usual variational problem, we would have to lift it to the jet bundle $J^1\tau_1$ (or perhaps to the subbundle $J^2\tau$). Instead, we will treat it as a Griffiths variational problem, in order to avoid the introduction of additional variables.

Regarding the constraint \eqref{eq:SkewSymmetricConnection}, let
\[
  \mathfrak{gl}\left(m\right)=\pf\oplus\kf
\]
be the decomposition of the general linear Lie algebra in terms of the $\pm1$-eigenspaces of the involutive operator
\[
  A\mapsto\eta A^T\eta;
\]
then this constraint is recovered by the equation
\begin{equation}\label{eq:GlobalMetricity}
  \rho^*\left(\pi_\pf\circ\theta_{J^1\tau}\right)=0,
\end{equation}
where $\pi_\pf:\mathfrak{gl}\left(m\right)\to\pf$ is the associated projection onto the factor $\pf$.

\begin{definition}[Variational problem for gravity with basis]\label{def:VarProblemPalatini}
  The \emph{variational problem for gravity with basis} is the triple
  \[
    \left(\tau_1:J^1\tau\to M,\lambda_{\text{PG}},\cI_{\text{PG}}\right)
  \]
  where $\cI_{\text{PG}}$ is the exterior differential system on $J^1\tau$ generated by the forms \eqref{eq:GlobalMetricity}.
\end{definition}

\section{Extension of generalized Cartan connections and Chern-Simons field theory}
\label{sec:extens-gener-cart-1}

In this section we will formulate a Chern-Simons field theory on a principal bundle with structure group $A\left(m\right)$, and we will relate it with the usual Chern-Simons field theory on a $K$-structure using the operation of extension for generalized Cartan connections.

\subsection{Chern-Simons variational problem with Lie group $A\left(3,\mR\right)$}
\label{sec:chern-simons-vari}

In the present section we will work with the affine frame bundle on a manifold of dimension $m=3$. Using Corollary \ref{cor:wise-vari-probl} and the geometrical constructions performed in Appendix \ref{sec:lift-conn-affine}, we will define a Griffiths variational problem on $\gamma^*\left(J^1\left(\tau\circ\beta\right)\right)$ in order to represent Chern-Simons gauge theory. To proceed, let us define the bilinear form $\left<\cdot,\cdot\right>:\mathfrak{gl}\left(m\right)\times\mathfrak{gl}\left(m\right)\to\mR$ given by
\[
  \left<a,b\right>:=a^k_ib^i_k.
\]
\begin{lemma}
  The bilinear form $\left<\cdot,\cdot\right>$ is non degenerate and $GL\left(3\right)$-invariant.
\end{lemma}
Recall that $K\subset GL\left(m,\mR\right)$ is the Lorentz group defined by the matrix $\eta$ (see Equation \eqref{eq:LorentzGroup}); as always, $\kf$ will indicate its Lie algebra. Then we have the isomorphism \cite{Wise_2010}
\[
  \kf\simeq\mR^3
\]
given by
\[
  \xi=\left(\xi^i\right)\mapsto a_i^j:=\eta^{jk}\epsilon_{ikl}\xi^l.
\]
It allows us to use the prescription
\[
  \left<\left(a,\xi\right),\left(b,\zeta\right)\right>:=\left<a,\zeta\right>+\left<b,\xi\right>
\]
for the extension of the bilinear form defined above to $\mathfrak{gl}\left(3\right)\ltimes\mR^3$. Thus we have a quadratic form
\[
  q:\mathfrak{a}\left(3\right)\to\mR:\left(a,\xi\right)\mapsto\left<\left(a,\xi\right),\left(a,\xi\right)\right>,
\]
and using Equation \eqref{eq:TransgressionBilinearForm}, we obtain the following definition.
\begin{definition}[Chern-Simons Lagrangian $3$-form]\label{def:chern-simons-vari}
  The \emph{Lagrangian form for Chern-Simons variational problem} is the $3$-form $\cL_{CS}\in \Omega^3\left(\gamma^*\left(J^1\left(\tau\circ\beta\right)\right)\right)$ defined through
  \[
    \cL_{CS}:=\left<\theta^*_{J^1\left(\tau\circ\beta\right)}\stackrel{\wedge}{,}\Theta^*_{J^1\left(\tau\circ\beta\right)}\right>-\frac{1}{6}\left<\theta^*_{J^1\left(\tau\circ\beta\right)}\stackrel{\wedge}{,}\left[\theta^*_{J^1\left(\tau\circ\beta\right)}\stackrel{\wedge}{,}\theta^*_{J^1\left(\tau\circ\beta\right)}\right]\right>,
  \]
  where $\theta^*_{J^1\left(\tau\circ\beta\right)}\in\Omega^1\left(\gamma^*\left(J^1\left(\tau\circ\beta\right)\right)\right)$ is the pullback to $\gamma^*\left(J^1\left(\tau\circ\beta\right)\right)$ of the canonical connection on the principal bundle $p_{A\left(3\right)}^{J^1\left(\tau\circ\beta\right)}:J^1\left(\tau\circ\beta\right)\to C\left(AM\right)$ (see Equation \eqref{eq:PullbackCanonicalConnection} above).
\end{definition}
Because $\theta^*_{J^1\left(\tau\circ\beta\right)}$ is $\mathfrak{gl}\left(3\right)\oplus\mR^3$-valued, we can write
\[
  \theta^*_{J^1\left(\tau\circ\beta\right)}=o^*_{J^1\tau}+e
\]
Using Equation \eqref{eq:LocalCoordinatesCanonicalFormAffine} we see that, in local terms, this form becomes
\[
  \theta^*_{J^1\left(\tau\circ\beta\right)}=e_\beta^j\left(de_i^\beta-e^\beta_{i\alpha}dx^\alpha\right)\otimes E^i_j-e^i_\beta v^\beta_{\alpha}dx^\alpha\otimes e_i,
\]
so that
\[
  o^*_{J^1\left(\tau\circ\beta\right)}=e_\beta^j\left(de_i^\beta-e^\beta_{i\alpha}dx^\alpha\right)\otimes E^i_j,\qquad e=-e^i_\beta v^\beta_{\alpha}dx^\alpha\otimes e_i.
\]
Also,
\[
 \Theta^*_{J^1\left(\tau\circ\beta\right)}=O^*_{J^1\left(\tau\circ\beta\right)}+E.
\]
Now, the splitting $\mathfrak{a}\left(3\right)=\gl\left(3\right)\oplus\mR^3$ has the following properties
\[
  \left[\gl\left(3\right),\gl\left(3\right)\right]\subset\gl\left(3\right),\qquad\left[\gl\left(3\right),\mR^3\right]\subset\mR^3,\qquad\left[\mR^3,\mR^3\right]=0
\]
and
\[
  \gl\left(3\right)\perp\gl\left(3\right),\qquad\mR^3\perp\mR^3,
\]
so that we can obtain a result that is the equivalent in this context to Proposition $1$ in \cite{wise2009symmetric}.
\begin{proposition}\label{prop:ChernSimonsRewrittenLagrangian}
  The Chern-Simons Lagrangian can be written as
  \[
    \cL_{\text{CS}}=\left<e\stackrel{\wedge}{,}O^*_{J^1\left(\tau\circ\beta\right)}\right>.
  \]
\end{proposition}
In local coordinates, this Lagrangian becomes
\[
  \cL_{\text{CS}}=\epsilon_{jkl}\eta^{kp}v^\nu_\beta e^j_\nu dx^\beta\wedge\left(O_{J^1\tau}^*\right)^l_p.
\]
It is interesting to note that, except by the factors $v^\alpha_\beta$, this Lagrangian is equivalent to the Lagrangian for $2+1$-gravity with basis; we need to take care of them, and it will be done through the imposition of constraints.

So, it is necessary to prescribe the set of constraints that sections $\sigma:U\subset M\to\gamma^*\left(J^1\left(\tau\circ\beta\right)\right)$ should obey in order to be evaluated in the action associated to $\cL_{CS}$. The first set of constraints we need to consider are those imposed by Equation \eqref{eq:SkewSymmetricConnection} above; in order to achieve it, let us consider the decomposition
\[
  \mathfrak{a}\left(3\right)=\left(\pf+0\right)\oplus\left(\kf\oplus\mR^3\right);
\]
accordingly, let
\[
  \pi_\pf:\mathfrak{a}\left(3\right)\to\pf+0,\qquad\pi_\kf:\mathfrak{a}\left(3\right)\to\kf\oplus\mR^3
\]
be the corresponding projectors. Then we will have that
\[
  \pi_\pf\circ\theta^*_{J^1\left(\tau\circ\beta\right)}=\eta^{ik}\left(o^*_{J^1\left(\tau\circ\beta\right)}\right)_k^j+\eta^{jk}\left(o^*_{J^1\left(\tau\circ\beta\right)}\right)_k^i,
\]
so that this set of forms is suitable for the incarnation of Equation \eqref{eq:SkewSymmetricConnection} in this setting. 

As we mentioned above, another constraint to be taken into account has to do with coordinates $v^\alpha_\beta$; the most natural thing is to use the map $j:J^1\tau\to J^1\left(\tau\circ\beta\right)$ discussed in Example \ref{example:AffineConnections}, where the classical notion for affine connection was introduced. 
In particular, it was proved there (see Equation \eqref{eq:jInLocalTerms}) that the image set of the map $j$ is described by the equations
\[
  v^\alpha=0,\quad v^\alpha_\beta=-e_\beta^ie_i^\alpha.
\]
Then the constraint in this case deals with the extra degrees of freedom $v^\alpha_\beta$ through this map; concretely, we are imposing the form
\[
  x\mapsto\left(\tau_{10}\left(s\left(x\right)\right),j\left(s\left(x\right)\right)\right)\in\gamma^*\left(J^1\left(\tau\circ\beta\right)\right)
\]
for the allowed sections in the variational problem with Lagrangian $\cL_{\text{CS}}$, where $s:U\subset M\to J^1\tau$ is a section of $\tau_1:J^1\tau\to M$.

\begin{definition}[Constraints for Chern-Simons variational problem on $A\left(3,\mR\right)$]
  We will say that a section $\sigma:U\subset M\to\gamma^*\left(J^1\left(\tau\circ\beta\right)\right)$ is \emph{admissible for the Chern-Simons variational problem} if and only if
  \[
    \sigma\left(x\right)=\left(\tau_{10}\left(s\left(x\right)\right),j\left(s\left(x\right)\right)\right)
  \]
  for some section $s:U\subset M\to J^1\tau$ and also
  \[
    \sigma^*\left(\pi_\pf\circ\theta^*_{J^1\left(\tau\circ\beta\right)}\right)=0.
  \]
\end{definition}
The admissible section will have the following property.
\begin{proposition}\label{prop:JMapProperties}
  Let $\theta_{J^1\left(\tau\circ\beta\right)}\in\Omega^1\left(J^1\left(\tau\circ\beta\right),\mathfrak{a}\left(m\right)\right)$ be the canonical connection on $J^1\left(\tau\circ\beta\right)$ and $\theta_{J^1\tau}\in\Omega^1\left(J^1\tau,\gl\left(m\right)\right)$ be the canonical connection form on $J^1\tau$. Then
  \[
    j^*\theta_{J^1\left(\tau\circ\beta\right)}=\theta_{J^1\tau}+\tau_{10}^*\varphi.
  \]
\end{proposition}
\begin{proof}
  We will give an argument in local terms. Using Equation \eqref{eq:LocalCoordinatesCanonicalFormAffine} and the local expression \eqref{eq:jInLocalTerms}, it follows that
\[
  j^*\theta_{J^1\left(\tau\circ\beta\right)}=e_\beta^j\left(de_i^\beta-e^\beta_{i\alpha}dx^\alpha\right)\otimes E^i_j+e^i_\beta\delta^\beta_{\alpha}dx^\alpha\otimes e_i=\theta_{J^1\tau}+\tau_{10}^*\varphi,
\]
as required.
\end{proof}

With these elements at hand, it is immediate to formulate the variational problem we will use to represent Chern-Simons gauge theory in this context.

\begin{definition}[Chern-Simons variational problem on $A\left(3,\mR\right)$]\label{def:ChernSimonsA3}
  It is the variational problem prescribed by the action
  \[
    \sigma\mapsto\int_U\sigma^*\left(\cL_{CS}\right)
  \]
  for $\sigma:U\subset M\to\gamma^*\left(J^1\left(\tau\circ\beta\right)\right)$ an admissible section.
\end{definition}

From Proposition \ref{prop:ChernSimonsRewrittenLagrangian} and \ref{prop:JMapProperties} we obtain the correspondence between Chern-Simons field theory and gravity with basis in this setting.

\begin{theorem}
  The extremals of the Chern-Simons variational problem on $A\left(3,\mR\right)$ are in one-to-one correspondence with the extremals of the variational problem described in Definition \ref{def:VarProblemPalatini}.
\end{theorem}
\begin{proof}
  The correspondence is determined by the bijective map
  \[
    s\mapsto\sigma:=\left(\tau_{10}\circ s,j\circ s\right)
  \]
  between a section $s$ of $\tau_1:J^1\tau\to M$ and an admissible section $\sigma$ for $\gamma^*\left(J^1\left(\tau\circ\beta\right)\right)$.
\end{proof}

\subsection{Gauge properties of the Chern-Simons Lagrangian $3$-form}
\label{sec:gauge-prop-chern}

In order to prove that this variational problem is able to reproduce the usual Chern-Simons field theory on a subbbundle, it will be necessary to prove that $\cL_{CS}$ established by Definition \ref{def:chern-simons-vari} has the correct transformation properties with respect to a gauge transformation (see Appendix \ref{sec:geom-princ-bundles}). To this end is devoted the following lemma, which is a reformulation of a previous result of Freed \cite{Freed:1992vw} to this new setting.

\begin{lemma}
  Let $\left(Q,\pi,N\right)$ be a $H$-principal bundle; indicate with $\theta\in\Omega^1\left(J^1\pi,\hf\right)$ the canonical connection on the bundle
  \[
    p_H^{J^1\pi}:J^1\pi\to C\left(Q\right).
  \]
  Also, let $s:U\subset M\to J^1\pi$ be a local section, and $g:U\to H$ a map. Define the new section
  \[
    \overline{s}:U\to J^1\pi:x\mapsto s\left(x\right)\cdot g\left(x\right).
  \]
  Then
  \[
    \left.\left(\overline{s}^*\theta\right)\right|_x=\text{Ad}_{g\left(x\right)}\circ\left.\left(s^*\theta\right)\right|_x+\left.\left(g^*\lambda\right)\right|_x,
  \]
  where $\lambda\in\Omega^1\left(H,\hf\right)$ is the (left) Maurer-Cartan $1$-form on $H$.
\end{lemma}

In the previous setting let us take $Q=AM$. By using the identification
\[
  LM\simeq\gamma\left(LM\right)\subset AM
\]
we can consider
\[
  \gamma^*\left(J^1\left(\tau\circ\beta\right)\right)\subset J^1\left(\tau\circ\beta\right);
\]
therefore, any section $s:U\subset\gamma^*\left(J^1\left(\tau\circ\beta\right)\right)$ can be seen as a section of the bundle $\left(\tau\circ\beta\right)_1:J^1\left(\tau\circ\beta\right)\to M$ taking values in this subbundle. With this in mind, we have the following consequence of the previous lemma.

\begin{corollary}\label{cor:GaugeInvarianceChernSimons}
  For $\cL_{CS}$ given by Definition \ref{def:chern-simons-vari} and a section $\overline{s}$ constructed as in the previous lemma, the following relation holds
  \[
    \overline{s}^*\cL_{CS}=s^*\cL_{CS}+d\left<\text{Ad}_{g^{-1}}\circ s^*\left(\theta^*\right)\stackrel{\wedge}{,}g^*\lambda\right>-\frac{1}{6}\left<g^*\lambda\stackrel{\wedge}{,}\left[g^*\lambda\stackrel{\wedge}{,}g^*\lambda\right]\right>.
  \]
\end{corollary}

\begin{remark}
  The $3$-form
  \[
    \left<g^*\lambda\stackrel{\wedge}{,}\left[g^*\lambda\stackrel{\wedge}{,}g^*\lambda\right]\right>
  \]
  is closed; therefore, the last term does not contribute to the action when performing variations, and so
  \[
    \delta\int_M\overline{s}^*\cL_{CS}=\delta\int_Ms^*\cL_{CS}.
  \]
\end{remark}

\subsection{Chern-Simons variational problem on $A\left(3\right)$ as extension of the Wise variational problem}
\label{sec:chern-simons-vari-as-Wise-problem}

We will prove in this section that sections of $J^1\tau$ that are extensions of extremals for the Wise variational problem on the jet bundle $J^1\pi_\zeta$ associated to any $K$-structure $O_\zeta$, are extremals for the Chern-Simons variational problem on $A\left(3\right)$ and viceversa. In order to achieve this result, we will need to use the relationship that connects extensions of Cartan connections with the original connections, as described in Diagram \eqref{eq:ConnectionExtensionDiagram}. In this case, with the help of the naturality of the canonical connection established in Lemma \ref{lem:CanonicalConnectionNatural}, the following proposition can be proven (the notation used is the one employed in Section \ref{sec:extens-gener-cart}).

\begin{proposition}\label{prop:CartanExtensionForms}
  Let
  \[
    \Gamma:P\to J^1\pi_{\left[G_1\right]}\qquad\Gamma^\sharp:P\left[G_2\right]\to J^1\pi_{\left[G\right]}
  \]
  be a pair of (generalized) Cartan connections such that
  \[
    j^1\gamma_{G_1}\circ\Gamma=\Gamma^\sharp\circ\gamma^2_{H}.
  \]
  Then
  \[
    \left(\Gamma^\sharp\circ\gamma^2_H\right)^*\theta_{J^1\pi_{\left[G\right]}}=\Gamma^*\theta_{J^1\pi_{\left[G_1\right]}}.
  \]
\end{proposition}

We will use this result to prove that the Chern-Simons variational problem on $A\left(3\right)$ can be seen as an extension of the Wise variational problem, namely, that we can establish a one-to-one correspondence between the extremals of these variational problems through the operations of extension and reduction of generalized Cartan connections, as defined in Section \ref{sec:extens-gener-cart}. To this end, let us apply Proposition \ref{prop:CartanExtensionForms} to the following diagram

\begin{equation*}
  \begin{tikzcd}[ampersand replacement=\&,row sep=.9cm,column sep=.9cm]
    \&
    AM
    \&
    \\
    O_\zeta^{\text{aff}}
    \arrow[hook]{ur}{\gamma^{\text{aff}}_\zeta}
    \&
    \&
    LM
    \arrow[hook',swap]{ul}{\gamma}
    \\
    \&
    O_\zeta
    \arrow[hook]{ul}{\gamma_\zeta}
    \arrow[hook',swap]{ur}{\gamma_{O_\zeta}}
  \end{tikzcd}      
\end{equation*}
Now, recall from Section \ref{sec:k-structures-space} that the choice of a metric $\zeta:M\to\Sigma$ allows us to select a $K$-structure $O_\zeta\subset LM$ and a subbundle $O_\zeta^{\text{aff}}\subset AM$; let us indicate with
\[
  \tau_\zeta^{\text{aff}}:O_\zeta^{\text{aff}}\to M
\]
the restriction of the canonical projection $\tau^{\text{aff}}:AM\to M$ to this subbundle. It follows that if
\[
 \sigma:U\to\gamma^*\left(J^1\left(\tau\circ\beta\right)\right)
\]
is a section of the jet space for the affine frame bundle, and
\[
  \sigma_\zeta:U\to\gamma_\zeta^*\left(J^1\left(\left.\left(\beta\circ\tau\right)\right|_{O_\zeta^{\text{aff}}}\right)\right)
\]
is a section of the jet space for the restriction of the affine frame bundle to orthonormal basis respect to the metric $\zeta$, then they will be related by a relation of reduction or extension if and only if
\[
  \tau_{10}\circ\text{pr}_1\circ\sigma=\left(\tau_\zeta^{\text{aff}}\right)_{10}\circ\text{pr}_1\circ\sigma_\zeta
\]
and
\[
  j^1\gamma_{\zeta}^{\text{aff}}\circ\text{pr}_2\circ\sigma=\text{pr}_2\circ\sigma_\zeta\circ\gamma_{O_\zeta},
\]
where $\text{pr}_i,i=1,2$ are the projections onto the factors in the cartesian product. Then by Proposition \ref{prop:CartanExtensionForms} we will have that
\[
  \sigma^*\cL_{\text{CS}}=\sigma_\zeta^*\lambda_{\text{CS}},
\]
for sections related by the operations of reduction or extension of generalized Cartan connections; it means that the correspondences
\[
  \sigma\mapsto\sigma_\zeta,\qquad\sigma_\zeta\mapsto\sigma
\]
establish a one-to-one correspondence between the Chern-Simons variational problem on $A\left(3\right)$ and the Wise variational problem, as required. Thus, we can prove the following result.

\begin{theorem}
  The operations of reduction and extension of generalized Cartan connections establish a one-to-one correspondence between the extremals of the Wise variational problem for any first order geometry associated to the pair $\left(SO\left(2,1\right)\ltimes\mR^3,SO\left(2,1\right)\right)$ (Definition \ref{def:WiseVariationalProblem}) and the Chern-Simons variational problem on $A\left(3\right)$ (as it is described by Definition \ref{def:ChernSimonsA3}).
\end{theorem}
\begin{proof}
  Let us suppose that we have a section
  \[
    \sigma:U\subset M\to\gamma^*\left(J^1\left(\tau\circ\beta\right)\right)\subset LM\times J^1\left(\tau\circ\beta\right);
  \]
  then, the induced section
  \[
    s_\sigma:=\text{pr}_1\circ\sigma
  \]
  has its image in a subbundle $O_\zeta\subset LM$ for some metric $\zeta$, which is given by the formula
  \[
    \zeta:=\eta^{ij}X_i\otimes X_j,
  \]
  where
  \[
    s_\sigma\left(x\right)=\left\{X_1\left(x\right),\cdots,X_m\left(x\right)\right\},\qquad x\in U.
  \]
  Recalling Remark \ref{rem:MetricityMeaning} and constraints \eqref{eq:GlobalMetricity}, we will have that its associated map $\Gamma_{A_{O_\zeta}}$ has its image in $J^1\left(\left.\left(\beta\circ\tau\right)\right|_{O^{\text{aff}}_\zeta}\right)$ and so its can be reduced to a connection $\sigma_\zeta$ for the subbundle $O_\zeta$.
  
  Now, the complicated part of the proof is the one that demonstrates that each extremal section of the variational problem in the restricted bundle is an extremal section for the general variational problem, because the variations of the first problem do not encompass all possible variations of the second problem. So, let us suppose that we have a section $\sigma_\lambda:U\to J^1\left(\left.\left(\beta\circ\tau\right)\right|_{O^{\text{aff}}_\zeta}\right)$ for the restricted variational problem; let $\sigma:U\subset M\to\gamma^*\left(J^1\left(\tau\circ\beta\right)\right)$ be a section for $\gamma^*\left(J^1\left(\tau\circ\beta\right)\right)$ induced by $\sigma_\zeta$. Consider $\sigma_t:U\subset M\to\gamma^*\left(J^1\left(\tau\circ\beta\right)\right)$ a variation for $\sigma$; then for every $t$ there exists a map $g_t:U\to A\left(3\right)$ such that the section
  \[
    \widetilde{\sigma}_t:=\sigma_t\cdot g_t
  \]
  verifies the condition $\mathop{\text{Im}}{\left(\text{pr}_1\circ\widetilde{\sigma}_t\right)}\subset O_\zeta$, and so it is induced by a variation $\overline{\sigma}_t$ for $\sigma_\zeta$. Then from Corollary \ref{cor:GaugeInvarianceChernSimons} we have that
  \[
    \overline{\sigma}_t^*d\lambda_{\text{CS}}=\widetilde{\sigma}_t^*d\cL_{\text{CS}}=\sigma_t^*d\cL_{\text{CS}},
  \]
  and therefore $\sigma$ is an extremal whenever $\sigma_\zeta$ is.
\end{proof}

\appendix

\section{Geometry of principal bundles}
\label{sec:geom-princ-bundles}

The following appendix contains the usual construction of a principal fiber bundle and its jet space from a local point of view. Although this construction is well-known, the article utilizes some of its consequences. Therefore, we have decided to include it here to establish the notation and provide the reader with a quick reference to these results.

\subsection{Connections and principal bundles}
\label{sec:conn-princ-bundl}

It is a usual to represent a principal connection by a family of locally defined $\g$-valued $1$-forms $A_U\in\Omega^1\left(U,\g\right)$, with $U$ belonging to a covering $\mathcal{C}:=\left\{U\right\}$ for $M$; these forms should obey a \emph{gauge transformation condition}: For every pair $U,V\in\mathcal{C}$ there must exists a map
\[
  t_{UV}:U\cap V\to G
\]
such that
\begin{equation}\label{eq:LocalGaugeTransformation}
  A_V=\text{Ad}_{t_{UV}^{-1}}\left(A_U\right)+t_{UV}^{-1}dt_{UV}=\text{Ad}_{t_{UV}^{-1}}\left(A_U\right)+t_{UV}^*\lambda,
\end{equation}
where $\lambda\in\Omega^1\left(G,\g\right)$ is the (left) Maurer-Cartan form on $G$. Additionally, these maps must be compatible in the sense that, for every $U,V,W\in\mathcal{C}$ such that $U\cap V\cap W\not=\emptyset$, then
\[
  t_{UW}\left(x\right)=t_{UV}\left(x\right)\cdot t_{VW}\left(x\right)
\]
for every $x\in U\cap V\cap W$. The existence of these maps is equivalent to have a $G$-principal bundle on $M$ \cite[p. 51]{KN1}, for which the family $\mathcal{C}$ becomes a covering of trivializing open sets; as it would seem obvious, we want to stress the fact that a principal bundle is singled out when dealing with a connection described in this way.

Before to go on, let us discuss briefly about a notion of equivalence involving the principal bundle structure. As we know, a gauge transformation does not change the connection; therefore, a collection of forms
\[
  A_U':=\text{Ad}_{g_U^{-1}}\left(A_U\right)+g_U^*\lambda
\]
where $g_U:U\to G$ is a family of smooth functions, should describe the same connection. Accordingly, the structure functions $t_{UV}$ must change obeying the rule
\begin{equation}\label{eq:EquivalentTransitionFunctions}
  t_{UV}':U\cap V\to G:x\mapsto g_U\left(x\right)t_{UV}\left(x\right)g_V\left(x\right),
\end{equation}
with the underlying principal bundle remaining invariant; therefore, set of transformations $\left\{t_{UV}\right\},\left\{t_{UV}'\right\}$ related by Equation \eqref{eq:EquivalentTransitionFunctions} should be considered as equivalent.

Let us look more closely to the principal bundle so constructed. Consider the triples
\[
  \left(U,x,a\right)\in\mathcal{C}\times M\times G
\]
such that $x\in U$; we say that $\left(U,x,a\right)$ is equivalent to $\left(V,y,b\right)$ if and only if $U\cap V\not=\emptyset$, $x=y$ and
\[
  b=t_{VU}\left(x\right)a.
\]
We will indicate with $\left[U,x,a\right]$ the equivalence class containing $\left(U,x,a\right)$; the space obtained by quotient by this equivalence relation becomes a $G$-principal bundle $\pi:P\to M$, where
\[
  \pi\left(\left[U,x,a\right]\right)=x
\]
and the $G$-action is simply given by
\[
  \left[U,x,a\right]\cdot h=\left[U,x,ah\right].
\]
The covering $\mathcal{C}$ contains trivializing open sets; in fact, on $\pi^{-1}\left(U\right)$ we have the trivializing map
\[
  t_U:\pi^{-1}\left(U\right)\to U\times G:\left[U,x,a\right]\mapsto\left(x,a\right).
\]
This description of the bundle $P$ allows us to construct a family of local sections for $P$, namely
\[
  s_U:U\to\pi^{-1}\left(U\right):x\mapsto\left[U,x,e\right],
\]
where we have used the symbol $e$ for the unit in $G$; it follows that
\begin{equation}\label{eq:SectionGaugeTransformation}
  s_V\left(x\right)=\left[V,x,e\right]=\left[U,x,t_{UV}\left(x\right)\right]=\left[U,x,e\right]\cdot t_{UV}\left(x\right)=s_U\left(x\right)\cdot t_{UV}\left(x\right)
\end{equation}
for every $x\in U\cap V$. It means that the local sections $s_U$ are related by the gauge transformations associated to the transition functions $t_{UV}$.

Let $U\in\mathcal{C}$ be an open set in the covering; then the pair $\left(A_U,s_U\right)$ allow us to construct the connection form $\omega$ on $P$ through the formula \cite{Naka}
\begin{equation}\label{eq:ConnectionFormFromLocalData}
  \left.\omega\right|_u:=\text{Ad}_{\left(g_U\left(u\right)^{-1}\right)}\left(\pi^*\left.A_U\right|_x\right)+g_U^*\lambda;
\end{equation}
here $x=\pi\left(u\right)$, $g_U:\pi^{-1}\left(U\right)\to G$ is defined by
\[
  u=s_U\left(x\right)\cdot g_U\left(u\right)
\]
and $\lambda$ is the (left) Maurer-Cartan form on $G$. The transformation properties of these local data imply that
\[
  \text{Ad}_{\left(g_U\left(u\right)^{-1}\right)}\left(\pi^*\left.A_U\right|_x\right)+g_U^*\lambda=\text{Ad}_{\left(g_V\left(u\right)^{-1}\right)}\left(\pi^*\left.A_V\right|_x\right)+g_V^*\lambda,\qquad x=\pi\left(u\right)
\]
for any pair $U,V\in\mathcal{C}$ and $u\in\pi^{-1}\left(U\cap V\right)$; therefore, this definition is independent of the open set used to calculate it through Eq. \eqref{eq:ConnectionFormFromLocalData}.

\subsection{Gauge transformations}
\label{sec:gauge-transf}

Recall that a \emph{gauge transformation of a principal bundle $P$} is a bundle map $\phi:P\to P$ over the identity that commutes with the $G$-action, namely
\[
  \phi\left(u\cdot g\right)=\phi\left(u\right)\cdot g.
\]
Using the trivialization maps $t_U:\pi^{-1}\left(U\right)\to U\times G$, a gauge transformation is locally described by the maps
\[
  \phi_U\left(x,a\right):=\left(t_U\circ\phi\circ t_U^{-1}\right)\left(x,a\right)=\left(x,h_U\left(x\right)a\right),
\]
where the functions $h_U:U\to G$ must have the following transformation property
\[
  h_V\left(x\right)=t_{UV}\left(x\right)h_U\left(x\right)
\]
for every $x\in U\cap V$. In fact, let us define
\[
  \phi\left(\left[U,x,a\right]\right):=\phi_U\left(x,a\right)=\left[U,x,h_U\left(x\right)a\right]
\]
for any $\left[U,x,a\right]\in P$; if $x\in U\cap V$, we will have that
\begin{align*}
  \left[V,x,h_V\left(x\right)a\right]&=\left[U,x,t_{VU}\left(x\right)h_V\left(x\right)a\right]\\
                                     &=\left[U,x,h_U\left(x\right)a\right],
\end{align*}
showing that this definition is independent of the open set $U$ containing $x$.

\subsection{Jet bundle and the bundle of principal connections}
\label{sec:bundle-princ-conn}

Let us suppose that every $U\in\mathcal{C}$ is a coordinate domain; therefore, we can write
\[
  A_U=\zeta^U_idx^i,\qquad\zeta^U_i:U\to\g
\]
and so, locally, we have that a connection can be seen as a section of the bundle
\[
  \text{pr}_1:U\times G\times\overbrace{\g\otimes\cdots\otimes\g}^{m\text{ times}}\to U.
\]
As with $P$, we can take advantage of the transition functions $t_{UV}$ in order to glue together these local fibrations. Accordingly, let us consider the $\left(m+3\right)$-uples
\[
  \left(U,x,a,\xi_1,\cdots,\xi_m\right)\in\mathcal{C}\times M\times G\times\overbrace{\g\otimes\cdots\otimes\g}^{m\text{ times}}
\]
such that $x\in U$, and define the equivalence relation given by
\[
  \left(U,x,a,\xi_i\right)\sim\left(V,y,b,\zeta_i\right)
\]
if and only if $U\cap V\not=\emptyset$, $x=y$, $b=t_{VU}\left(x\right)a$ and
\[
  \zeta_idx^i=\xi_idx^i+\text{Ad}_{a^{-1}}\left(t_{VU}^*\lambda\right).
\]
The quotient space $J^1\pi$ is a manifold, the so called \emph{jet space of the bundle $\pi:P\to M$}; we have canonical projections
\[
  \pi_{10}:J^1\pi\to P:\left[U,x,a,\xi_i\right]\mapsto\left[U,x,a\right]
\]
and
\[
  \pi_1:J^1\pi\to M:\left[U,x,a,\xi_i\right]\mapsto x
\]
giving it bundle structure on both spaces $P$ and $M$. As before, we have a local trivialization for $J^1\pi$ through the formula
\[
  T_U:\pi_1^{-1}\left(U\right)\to U\times G\times\overbrace{\g\otimes\cdots\otimes\g}^{m\text{ times}}:\left[U,x,a,\xi_i\right]\mapsto\left(x,a,\xi_i\right).
\]
Every element $\left[U,x,g,\xi_i\right]\in J^1\pi$ is equivalent to a linear map
\[
  \frac{\partial}{\partial x^i}\in T_xM\longmapsto T_{\left[U,x,e\right]}R_g\circ T_xs_U\left(\frac{\partial}{\partial x^i}\right)+\left(\xi_i\right)_P\left(\left[U,x,g\right]\right)\in T_{\left[U,x,g\right]}P,
\]
where $\zeta_P$ is the infinitesimal generator of the $G$-action on $P$ corresponding to the element $\zeta\in\g$. Using the fact that
\[
  t_U\circ R_h=\left(\text{id}\times R_h\right)\circ t_U
\]
for every $h\in G$, we can see that
\[
  Tt_U\circ\zeta_P=\zeta_G^R
\]
for all $\zeta\in\g$; here $\xi_G^R$ indicates the infinitesimal generator on $G$ associated to the right action. Therefore, the map $T_U$ is induced by $Tt_U$.

It can be seen that $J^1\pi$ is a $G$-space; the action is given by the formula
\[
  \left[U,x,a,\xi_i\right]\cdot h=\left[U,x,ah,\text{Ad}_{h^{-1}}\xi_i\right]
\]
for all $h\in G$. This action allows us to construct a quotient bundle
\[
  C\left(P\right):=J^1\pi/G,
\]
and the canonical projection gives rise to a new bundle
\[
  p_G^{J^1\pi}:J^1\pi\to C\left(P\right).
\]
The map $\overline{\pi}:C\left(P\right)\to M$ such that
\[
  \pi_1=\overline{\pi}\circ p_G^{J^1\pi}
\]
gives $C\left(P\right)$ structure of a bundle on $M$; its elements are the equivalence classes
\[
  \left[U,x,a,\xi_i\right]_G:=\left\{\left[U,x,ah,\text{Ad}_{h^{-1}}\xi_i\right]:h\in G\right\}.
\]
For every $U\in\mathcal{C}$, we have a trivialization map
\[
  \phi_U:\overline{\pi}^{-1}\left(U\right)\to U\times\g^{\otimes m}:\left[U,x,a,\xi_i\right]_G\mapsto\left(x,\text{Ad}_a\xi_i\right);
\]
for every pair $U,V\in\mathcal{C}$ such that $U\cap V\not=\emptyset$, we have that
\begin{align*}
  \phi_V\circ\phi_U^{-1}\left(x,\xi_i\right)&=\phi_V\left(\left[U,x,e,\xi_i\right]_G\right)\\
                                            &=\phi_V\left(\left[V,x,t_{VU}\left(x\right),\xi_i+\left(t_{VU}^*\lambda\right)\left(\frac{\partial}{\partial x^i}\right)\right]_G\right)\\
                                            &=\left(x,\text{Ad}_{t_{VU}\left(x\right)}\left[\xi_i+\left(t_{VU}^*\lambda\right)\left(\frac{\partial}{\partial x^i}\right)\right]\right)\\
  &=\left(x,\text{Ad}_{t_{UV}\left(x\right)^{-1}}\xi_i-\left(t_{UV}^*\lambda\right)\left(\frac{\partial}{\partial x^i}\right)\right)
\end{align*}
where it was used that $t_{VU}\left(x\right)=t_{UV}\left(x\right)^{-1}$. The fact that this expression is equivalent to the transformation law \eqref{eq:LocalGaugeTransformation} allows us to consider $C\left(P\right)$ as the bundle of principal connections for $P$; namely, we have an identification
\begin{equation}\label{eq:LocalFormCorrespondence}
  \left[U,x,a,\xi_i\right]_G\longleftrightarrow A_U:=-\text{Ad}_a\xi_idx^i.
\end{equation}
With this correspondence in mind, the existence of a $\g$-valued $1$-form $A_U$ for every $U\in\mathcal{C}$ such that compatibility conditions \eqref{eq:LocalGaugeTransformation} are fulfilled, is equivalent to the existence of a section
\[
  \sigma:M\to C\left(P\right).
\]
In fact, we have the formula
\[
  \sigma\left(x\right)=\left[U,x,e,-\left.A_U\right|_x\left(\frac{\partial}{\partial x^i}\right)\right]_G
\]
for every $U\in\mathcal{C}$ and $x\in U$.

\subsection{The canonical connection form on $J^1\pi$}
\label{sec:canon-conn-form}

A fundamental geometric structure on $J^1\pi$ is the \emph{canonical connection form} $A$, which is a $\g$-valued $1$-form on $J^1\pi$ inducing a connection on the $G$-principal bundle $p_G^{J^1\pi}:J^1\pi\to C\left(P\right)$; for every $U\in\mathcal{C}$, it is given by the formula
\[
  \left.A\right|_{\left[U,x,a,\xi_i\right]}:=\left.\lambda\right|_a-\xi_idx^i.
\]
In fact, because $b=t_{VU}\left(x\right)a$, it results that
\[
  \left.\lambda\right|_b=\left.\lambda\right|_a+\text{Ad}_{a^{-1}}\left(t_{VU}^*\lambda\right);
\]
thus, if $\left[U,x,a,\xi_i\right]=\left[V,x,b,\zeta_i\right]$ for $x\in U\cap V$, we will have
\begin{align*}
  \left.A\right|_{\left[V,x,b,\zeta_i\right]}&=\left.\lambda\right|_{b}-\zeta_idx^i\\
                                             &=\left.\lambda\right|_a+\text{Ad}_{a^{-1}}\left(t_{VU}^*\lambda\right)-\left(\xi_idx^i+\text{Ad}_{a^{-1}}\left(t_{VU}^*\lambda\right)\right)\\
                                             &=\left.\lambda\right|_a-\xi_idx^i\\
  &=\left.A\right|_{\left[U,x,a,\xi_i\right]},
\end{align*}
and the form $A$ is well-defined.

There is an important property that the canonical form has. In order to formulate it, let us consider another way to specify a connection, namely, through an equivariant bundle map
\[
  \Gamma:P\to J^1\pi.
\]
In fact, using the above description of these bundles, and given a local description $\left\{A_U:U\in\mathcal{C}\right\}$ for a connection, we can construct the map
\[
  \Gamma:\left[U,x,a\right]\mapsto\left[U,x,a,-\text{Ad}_{a^{-1}}\left(\left.A_U\right|_x\left(\frac{\partial}{\partial x^i}\right)\right)\right],
\]
where identification \eqref{eq:LocalFormCorrespondence} was used. Having the gauge transformation property \eqref{eq:LocalGaugeTransformation} in mind, we can prove that it is a good definition. Now, given that
\[
  s_U\left(x\right)=\left[U,x,e\right]
\]
and from $u=s_U\left(x\right)\cdot g_U\left(u\right)$ for every $u=\left[U,x,a\right]\in\pi^{-1}\left(U\right)$, we obtain the formula
\[
  u=\left[U,x,g_U\left(u\right)\right].
\]
Then
\[
  \Gamma^*\left(\left.A\right|_{\Gamma\left(u\right)}\right)=\Gamma^*\left(\lambda-\xi_idx^i\right)=g_U^*\lambda+\text{Ad}_{g_U\left(u\right)^{-1}}\left.A_U\right|_x=\left.\omega\right|_u,
\]
namely, the connection form can be obtained through pullback along the map $\Gamma$ of the canonical form.

\section{Local and global Chern-Simons Lagrangians}
\label{sec:local-global-chern}

Let $K$ be a Lie group and $\pi_\zeta:R_\zeta\to M$ a $K$-principal bundle (the notation will be explained later); on its first order jet space
\[
  \left(\pi_\zeta\right)_1:J^1\pi_\zeta\to M
\]
we will define a (global) variational problem, which we will prove  to represent Chern-Simons gauge theory. In order to accomplish this task, it will be necessary to lift it to $J^1\left(\pi_{\zeta}\right)_1$, and compare it with the variational problem defined by local data, which lives on $J^1\overline{\pi}_\zeta$.

Let us now suppose that we have an invariant polynomial $q:\kf\to\mR$ of degree $n$. According to the Chern-Simons theory \cite{10.2307/1971013,Morita1}, the $2n$-form
\[
  \alpha:=q\left(F\right)\in\Omega^{2n}\left(J^1\pi_\zeta\right)
\]
is closed. For example, when the Lie algebra comes with an invariant bilinear form, we can consider the quadratic polynomial
  \[
    q\left(F\right):=\left<F\stackrel{\wedge}{,}F\right>;
  \]
Additionally, it can be proved that the bundle
\[
  \text{pr}_1:J^1\pi_\zeta\times_{C\left( {R}_\zeta\right)} J^1\pi_\zeta\to J^1\pi_\zeta,
\]
is a trivial $K$-principal bundle; it means that $\alpha$ is not only closed, but also exact. Therefore, there exists a $\left(2n-1\right)$-form
\begin{equation}\label{eq:ChernSimonsLagrangianBeta}
  \beta:=Tq\left(A,F\right)\in\Omega^{2n-1}\left(J^1\pi_\zeta\right),
\end{equation}
such that
\[
  d\beta=\text{pr}_1^*\alpha;
\]
the polynomial $Tq$ can be found by a transgression formula \cite{Naka}. When $q$ is the quadratic form determined by an invariant bilinear form, the transgression is given by
\begin{equation}\label{eq:TransgressionBilinearForm}
  Tq\left(A,F\right)=\left<A\stackrel{\wedge}{,}F\right>-\frac{1}{6}\left<A\stackrel{\wedge}{,}\left[A\stackrel{\wedge}{,}A\right]\right>.
\end{equation}

How can these forms be related to the local forms usually used to described Chern-Simons field theory? In this case, we have another $K$-principal bundle structure, namely the quotient map
\[
  p_{K}^{J^1\pi_\zeta}:J^1\pi_\zeta\to C\left( {R}_\zeta\right),
\]
and $q\left(F\right)$ defines a Chern class for it. Accordingly, there exists a $2n$-form $\gamma$ on $C\left( {R}_\zeta\right)$ such that
\[
\left(p_{K}^{J^1\pi_\zeta}\right)^*\gamma=q\left(F\right).
\]
Moreover, using the canonical $2$-form $F_2$ on the bundle of connections $C\left( {R}_\zeta\right)$, we can prove that
\[
  \gamma=q\left(F_2\right).
\]
But now, this bundle is not trivial in general; in short, from decomposition
\[
  J^1\pi_\zeta= {R}_\zeta\times_M C\left( {R}_\zeta\right),
\]
we obtain that it is trivial if and only if the bundle
\[
  \pi_\zeta: {R}_\zeta\to M
\]
is. In consequence, it is not expected that the $2n$-form $\gamma$, although closed, should also be exact; it is the reason why, although we have a global variational problem on $J^1\pi_\zeta$, it cannot be reproduced on $C\left( {R}_\zeta\right)$, even having in mind that the transformations properties of $\cL_{CS}$ are telling us that the degrees of freedom associated to the $ {R}_\zeta$-factor can be ignored.

Nevertheless, because $ {R_\zeta}$ admits trivializing open sets, the previous considerations can be used to associate to every such set $U\subset M$ a Lagrangian $\cL_U:=Tq\left(A_2^U,F_2\right)$, where
\[
  A_2^U\in\Omega^1\left(\left(\overline{\pi}_\zeta\right)^{-1}\left(U\right),\kf\right)
\]
is a $1$-form such that
\[
  \left(p_{K}^{J^1\pi_\zeta}\right)^*A_2^U=\left.A\right|_{\left(\left(\pi_\zeta\right)_1\right)^{-1}\left(U\right)}.
\]
\begin{remark}
  A study of conditions ensuring the existence of global solutions for the local  variational problem for Chern-Simons gauge theory can be found in \cite{Palese2017}. In this regard, it is interesting to note that here we have changed a variational problem determined by local data and whose sections could be globally defined, by a variational problem described by global data, but whose sections are forced to have local nature (part of these sections are sections of a principal bundle).
\end{remark}

\section{Geometry of the affine frame bundle}
\label{sec:lift-conn-affine}

From now on we will devoted ourselves to particularize this definition to a very specific principal bundle, the so called \emph{affine frame bundle} (see Definition \ref{def:Affine-Frame-Bundle} below), and to relate the variational problem so obtained with Palatini gravity.

\subsection{The affine general linear group and the affine frame bundle}
\label{sec:affine-gener-line}

We have the splitting short exact sequence
\[
  \begin{tikzcd}[ampersand replacement=\&,row sep=1.5cm,column sep=1.5cm]
    0
    \arrow{r}{}
    \&
    \mathbb{R}^m
    \arrow{r}{\alpha}
    \&
    A\left(m,\mathbb{R}\right)
    \arrow{r}{\beta}
    \&
    GL\left(m,\mathbb{R}\right)
    \arrow[swap]{r}{}
    \arrow[bend left=45,dashed]{l}{\gamma}
    \&
    1
  \end{tikzcd}  
\]
where $A\left(m,\mathbb{R}\right)\subset GL\left(m+1,\mathbb{R}\right)$ is the subgroup of matrices of the form
\[
  B:=
  \begin{bmatrix}
    a&\xi\\
    0&1
  \end{bmatrix}
\]
where $a\in GL\left(m,\mathbb{R}\right)$ and $\xi\in\mR^m$; the maps in the sequence read
\[
  \alpha\left(\xi\right):=\begin{bmatrix}
    1&\xi\\
    0&1
  \end{bmatrix},\qquad
  \beta\left(\begin{bmatrix}
      a&\xi\\
      0&1
    \end{bmatrix}\right):=
  \begin{bmatrix}
    a&0\\
    0&1
  \end{bmatrix},
\]
and
\[
  \gamma\left(a\right):=\begin{bmatrix}
    a&0\\
    0&1
  \end{bmatrix}.
\]
Because this short sequence splits, we can consider
\[
  A\left(m,\mathbb{R}\right)=GL\left(m,\mathbb{R}\right)\oplus\mR^m,
\]
with the isomorphism of groups given by
\[
  GL\left(m,\mathbb{R}\right)\oplus\mR^m\to A\left(m,\mathbb{R}\right):\left(a,\xi\right)\mapsto\gamma\left(a\right)+\alpha\left(\xi\right).
\]
Now, let $A^m$ be the set $\mR^m$ considered as an affine space; we can set an isomorphism between $A\left(m,\mathbb{R}\right)$ and the set of affine maps
\[
  f:A^m\to A^m;
\]
in fact, given $B=\left(a,\xi\right)$, the associated affine map reads
\[
  f_B\left(z\right):=az+\xi
\]
for every $z\in A^m$. Using the expression of an element $\left(a,\xi\right)$ as a matrix, we have that
\[
  \left(a,\xi\right)^{-1}=\left(a^{-1},-a^{-1}\xi\right)
\]
and
\begin{equation}\label{eq:AdjointActionSemidirect}
  \text{Ad}_{\left(a,\xi\right)}\left(b,\zeta\right)=\left(\text{Ad}_ab,a\zeta-\left(\text{Ad}_ab\right)\xi\right)
\end{equation}
for any $\left(a,\xi\right)\in A\left(m,\mR\right),\left(b,\zeta\right)\in\mathfrak{a}\left(m,\mR\right)$.

In the same vein, let $A_x\left(M\right)$ be the set $T_xM$ considered as an affine space, for every $x\in M$. As it is well-known \cite{KN1}, the set of affine maps
\[
  u:A^m\to A_x\left(M\right)
\]
for every $x\in M$ has structure of $A\left(m,\mathbb{R}\right)$-principal bundle; the action of an element $B\in A\left(m,\mathbb{R}\right)$ is simply given by 
\[
  u\cdot B:=u\circ B.
\]

\begin{definition}[Bundle of affine frames]\label{def:Affine-Frame-Bundle}
  The \emph{bundle of affine frames on $M$} will be the set
  \[
    AM:=\bigcup_{x\in M}\left\{u:A^m\to A_x\left(M\right)\text{ affine}\right\}.
  \]
\end{definition}

As it follows from the general theory of principal bundles, there exists a pair of principal bundle morphisms associated to the homomorphisms $\beta:A\left(m,\mathbb{R}\right)\to GL\left(m,\mathbb{R}\right)$ and $\gamma:GL\left(m,\mathbb{R}\right)\to A\left(m,\mathbb{R}\right)$
\[
  \begin{tikzcd}[ampersand replacement=\&,row sep=1.5cm,column sep=1.5cm]
    AM
    \arrow[shift right,swap]{rr}{\beta}
    \arrow{dr}{}
    \&
    \&
    LM
    \arrow[shift right,swap]{ll}{\gamma}
    \arrow{dl}{\tau}
    \\
    \&
    M
    \&
  \end{tikzcd}  
\]
For any $G$-principal bundle $\pi:P\to M$, the affine bundle $\left(C\left(P\right),\overline{\pi},M\right)$ defined through the diagram
\[
  \begin{tikzcd}[ampersand replacement=\&,row sep=1.5cm,column sep=1.5cm]
    J^1\pi
    \arrow{r}{\pi_{10}}
    \arrow[swap]{d}{p_G^{J^1\pi}}
    \&
    P
    \arrow{d}{\pi}
    \\
    C\left(P\right):=J^1\pi/G
    \arrow{r}{\overline{\pi}}
    \&
    M
  \end{tikzcd}  
\]
is called the \emph{bundle of connections of the bundle $P$}, and we can establish a canonical one-to-one correspondence between its sections and principal connections on $P$. The correspondence is given as follows: Any element $j_x^1s\in J^1\pi$ is a linear map
\[
  j_x^1s:T_xM\to T_{s\left(x\right)}P
\]
such that
\[
  T_{s\left(x\right)}\pi\circ j_x^1s=\text{id}_{T_xM},
\]
and so a $G$-orbit $\left[j_x^1s\right]_G$ can be interpreted as a linear map
\[
  \left[j_x^1s\right]_G:T_xM\to\left(TP/G\right)_x.
\]
Given $u\in P$, there exists a unique $m$-dimensional subspace $H_u\subset T_uP$ such that
\[
  \left[j_x^1s\right]_G\left(T_xM\right)=p_G^{TP}\left(H_u\right);
\]
the assignment $u\mapsto H_u$ is the connection associated to $\left[j_x^1s\right]_G$.

Therefore we have the diagram
\[
  \begin{tikzcd}[ampersand replacement=\&,row sep=1.5cm,column sep=1.5cm]
    J^1\left(\tau\circ\beta\right)
    \arrow{rr}{p_{A\left(m,\mathbb{R}\right)}^{J^1\left(\tau\circ\beta\right)}}
    \arrow{dr}{j^1\beta}
    \arrow[swap]{dd}{\left(\tau\circ\beta\right)_{10}}
    \&
    \&
    C\left(AM\right)
    \arrow{d}{\left[j^1\beta\right]}
    \\
    \&
    J^1\tau
    \arrow{r}{p_{GL\left(m,\mathbb{R}\right)}^{J^1\tau}}
    \arrow{d}{\tau_{10}}
    \&
    C\left(LM\right)
    \arrow{d}{\overline{\tau}}
    \\
    AM
    \arrow{r}{\beta}
    \&
    LM
    \arrow{r}{\tau}
    \&
    M
  \end{tikzcd}  
\]
It is a theorem that any connection $\Gamma:M\to C\left(LM\right)$ gives rise to a unique connection $\widetilde{\Gamma}:M\to C\left(AM\right)$ such that if
\[
  \omega_{LM}\in\Omega^1\left(LM,\mathfrak{gl}\left(m,\mathbb{R}\right)\right)\qquad\text{ and }\qquad{\omega}_{AM}\in\Omega^1\left(AM,\mathfrak{a}\left(m,\mathbb{R}\right)\right)
\]
are the corresponding connection forms, then
\begin{equation}\label{eq:AffineVsLinearConnections}
  \gamma^*{\omega}_{AM}=\omega_{LM}+\varphi,
\end{equation}
where $\varphi\in\Omega^1\left(LM,\mathbb{R}^m\right)$ is the canonical solder $1$-form on $LM$.

\subsection{The affine frame bundle as extension of the frame bundle}
\label{sec:affine-frame-bundle}

It remains to interpret the affine frame bundle as the extension of the frame bundle using the group immersion $GL\left(m,\mR\right)\subset A\left(m,\mR\right)$.
\begin{proposition}\label{prop:FrameBundleExtension}
  Let $H=GL\left(m,\mR\right)$, $G=A\left(m,\mR\right)$ and $P=LM$. Then
  \[
    P\times_HG\simeq AM.
  \]
\end{proposition}
\begin{proof}
  Let us define the map
  \[
    \phi\left(\left[u,\left(h,v\right)\right]_H\right):=a\in\left.AM\right|_x
  \]
  if and only if $x=\tau\left(u\right)$ and
  \[
    a:\mR^m\to T_xM:w\mapsto\left(u\circ h\right)\left(w\right)+u\left(v\right).
  \]
  Then, for any $\left(h',v'\right)\in A\left(m,\mR\right)$, we have that
  \begin{align*}
    \phi\left(\left[u,\left(h,v\right)\right]_H\cdot\left(h',v'\right)\right)&=\phi\left(\left[u,\left(h,v\right)\left(h',v'\right)\right]_H\right)\\
                                                                             &=\phi\left(\left[u,\left(hh',hv'+v\right)\right]_H\right)\\
                                                                             &=\left[w\mapsto\left(u\circ h\circ h'\right)\left(w\right)+u\left(hv'+v\right)\right]\\
                                                                             &=\left[w\mapsto u\left(h\left(h'w+v'\right)\right)+u\left(v\right)\right]\\
    &=\phi\left(\left[u,\left(h,v\right)\right]_H\right)\circ\left(h',v'\right),
  \end{align*}
  proving that $\phi$ is a bundle map.
  
  For every $a\in\left.AM\right|_x$, we have that
  \[
    \phi\left(\left[u,\left(e,v\right)\right]_H\right)=a
  \]
  if and only if $v=a\left(0\right)$ and
  \[
    u\left(w\right)=a\left(w\right)-a\left(0\right)
  \]
  for all $w\in\mR^m$; therefore, $\phi$ es an epimorphism of bundles.

  Additionally, if $\left[u_1,\left(h_1,v_1\right)\right]_H,\left[u_2,\left(h_2,v_2\right)\right]_H$ are such that
  \[
    \phi\left(\left[u_1,\left(h_1,v_1\right)\right]_H\right)=\phi\left(\left[u_2,\left(h_2,v_2\right)\right]_H\right),
  \]
  then
  \[
    \left(u_1\circ h_1\right)\left(w\right)+u_1\left(v_1\right)=\left(u_2\circ h_2\right)\left(w\right)+u_2\left(v_2\right)
  \]
  for all $w\in\mR^m$. With $w=0$ it gives us that
  \begin{equation}\label{eq:SemidirectEquation}
    u_1\left(v_1\right)=u_2\left(v_2\right)
  \end{equation}
  and so
  \[
    u_2\circ h_2=u_1\circ h_1\qquad\Longrightarrow\qquad u_2=u_1\circ h_1\circ h_2^{-1}.
  \]
  Thus, Equation \eqref{eq:SemidirectEquation} tells us that
  \[
    v_2=\left(h_2\circ h_1^{-1}\right)\left(v_1\right)
  \]
  and it means that
  \[
    \left[u_2,\left(h_2,v_2\right)\right]_H=\left[u_2\circ h_1\circ h_2^{-1},\left(h_2,\left(h_2\circ h_1^{-1}\right)\left(v_1\right)\right)\right]_H=\left[u_1,\left(h_1,v_1\right)\right]_H.
  \]
  Namely, $\phi$ is a monomorphism of bundles.
\end{proof}

\subsection{Local expressions}
\label{sec:local-expressions}

We will use the constructions developed in Section \ref{sec:geom-princ-bundles} in order to find coordinates for $LM$, $AM$ and its jet bundles. The first thing to note is that the frame bundle $LM$ can be trivialized on every coordinate chart $\left(U,\phi\right)$ for $M$; namely, for $u\in\tau^{-1}\left(U\right)$ there exists a collection $\left(e_i^\alpha\left(u\right)\right)$ of real numbers such that
\[
  u\left(c^1,\cdots,c^m\right)=c^ie_i^\alpha\left(u\right)\frac{\partial}{\partial x^\alpha},\qquad\left(c^1,\cdots,c^m\right)\in\mR^m,
\]
where $\phi=\left(x^\alpha\right)$ are the coordinate functions on $\phi\left(U\right)\subset\mR^m$. It induces the coordinate chart on $\tau^{-1}\left(U\right)$ given by
\[
  \Phi_U\left(u\right):=\left(x^\alpha\left(\tau\left(u\right)\right),e_i^\alpha\left(u\right)\right).
\]
In the same vein, given $\overline{u}\in\left(\tau\circ\beta\right)^{-1}\left(U\right)\subset AM$, we can find numbers $\left(e_i^\alpha\left(\overline{u}\right),v^\alpha\left(\overline{u}\right)\right)$ such that
\begin{equation}\label{eq:AffineBundleLocal}
  \overline{u}\left(c^1,\cdots,c^m\right)=\left[c^ie_i^\alpha\left(\overline{u}\right)+v^\alpha\left(\overline{u}\right)\right]\frac{\partial}{\partial x^\alpha},\qquad\left(c^1,\cdots,c^m\right)\in\mR^m;
\end{equation}
it defines a coordinate chart on $\left(\tau\circ\beta\right)^{-1}\left(U\right)$ through the formula
\[
 \overline{\Phi}_U\left(\overline{u}\right):=\left(x^\alpha\left(\tau\left(\beta\left(\overline{u}\right)\right)\right),e_i^\alpha\left(\overline{u}\right),v^\alpha\left(\overline{u}\right)\right).
\]
The map $\gamma$ fits nicely with these coordinates; in fact, we have that
\[
  \overline{\Phi}_U\circ\gamma\circ\Phi_U^{-1}\left(x^\alpha,e^\beta_i\right)=\left(x^\alpha,e^\beta_i,0\right).
\]
Let us consider now the action of an element $\left(a,w\right)\in A\left(m,\mR\right)=GL\left(m,\mR\right)\oplus\mR^m$ on $AM$; because of the equation \eqref{eq:AffineBundleLocal}, we have that
\begin{align*}
  \left[\overline{u}\cdot\left(a,w\right)\right]\left(c^1,\cdots,c^m\right)&=\overline{u}\left(a_j^1c^j+w^1,\cdots,a_j^mc^j+w^m\right)\\
                                                                           &=\left[\left(a_j^ic^j+w^i\right)e_i^\alpha\left(\overline{u}\right)+v^\alpha\left(\overline{u}\right)\right]\frac{\partial}{\partial x^\alpha}\\
  &=\left[c^ie_i^\alpha\left(\overline{u}\right)+\left(v^\alpha\left(\overline{u}\right)+c^ie_i^\alpha\left(\overline{u}\right)\right)\right]\frac{\partial}{\partial x^\alpha},
\end{align*}
namely,
\[
  \left(x^\alpha,e_i^\beta,v^\rho\right)\cdot\left(a,w\right)=\left(x^\alpha,a^i_je_i^\beta,v^\rho+e^\rho_iw^i\right).
\]

We are ready to deal with connections on $AM$ in a local fashion; according to the discussion carried out in Section \ref{sec:conn-princ-bundl}, the local version of a connection on $AM$ is an $\mathfrak{a}\left(m,\mR\right)$-valued $1$-form, namely
\[
  \widetilde{\omega}_U=\left(\Gamma^\alpha_{\gamma\beta}dx^\beta\otimes E_\alpha^\gamma,\sigma^\alpha_\beta dx^\beta\otimes e_\alpha\right).
\]
Accordingly, it can be globalized through formula \eqref{eq:ConnectionFormFromLocalData}; first, recall that in this case we have
\[
  g_U\left(x^i,e^\alpha_i,v^\beta\right)=\left(e^\alpha_i,v^\beta\right),
\]
and so
\[
  g_U^*\lambda=\left(e^\alpha_i,v^\beta\right)^{-1}\left(de^\alpha_i,dv^\beta\right)=\left(e^i_\alpha de^\alpha_j,e^i_\beta dv^\beta\right).
\]
Additionally, the adjoint action formula \eqref{eq:AdjointActionSemidirect} tells us that
\[
 \text{Ad}_{\left(g_U\left(x^i,e^\alpha_i,v^\beta\right)\right)^{-1}}\widetilde{\omega}_U=\left(e^\gamma_ie^j_\alpha\Gamma^\alpha_{\gamma\beta}dx^\beta\otimes E_j^i,e^i_\alpha\left(\Gamma^\alpha_{\gamma\beta}v^\gamma+\sigma^\alpha_\beta\right)dx^\beta\otimes e_i\right),
\]
so that
\begin{equation}\label{eq:AffineConnectionGlobal}
  \left.\widetilde{\omega}\right|_{\overline{\Phi}_U\left(x^i,e^\beta_j,v^\gamma\right)}=\left(e^j_\alpha\left(de^\alpha_i+e^\gamma_i\Gamma^\alpha_{\gamma\beta}dx^\beta\right)\otimes E_j^i,e^i_\alpha\left[dv^\alpha+\left(\Gamma^\alpha_{\gamma\beta}v^\gamma+\sigma^\alpha_\beta\right)dx^\beta\right]\otimes e_i\right).
\end{equation}
Coordinates $\Phi_U$ and $\overline{\Phi}_U$ induce coordinates on $J^1\tau$ and $J^1\left(\tau\circ\beta\right)$ respectively, which will be indicated as
\[
 \left(x^\alpha,e^\beta_i,e^\beta_{i\gamma}\right)\qquad\text{and}\qquad\left(x^\alpha,e^\beta_i,v^\alpha,e^\beta_{i\gamma},v^\alpha_\beta\right).
\]
The $A\left(m,\mR\right)$-action lifts to $J^1\left(\tau\circ\beta\right)$ as follows
\[
  \left(x^\alpha,e_i^\beta,v^\rho,e^\beta_{i\gamma},v^\alpha_\beta\right)\cdot\left(a,w\right)=\left(x^\alpha,a^i_je_i^\beta,v^\rho+e^\rho_iw^i,a_j^ie^\beta_{i\gamma},v^\alpha_\beta+e^\alpha_{i\beta}w^i\right).
\]
using the form of this action, we can consider the projection from $J^1\left(\tau\circ\beta\right)$ to $C\left(AM\right)$; we have that
\[
 p_{A\left(m,\mR\right)}^{J^1\left(\tau\circ\beta\right)}\left(x^\alpha,e_i^\beta,v^\rho,e^\beta_{i\gamma},v^\alpha_\beta\right)=\left(x^i,e^j_\gamma e^\alpha_{j\beta},v^\alpha_\beta-e^i_\gamma e^\alpha_{i\beta}v^\gamma\right).
\]
Let us now consider the canonical connection $\theta_{J^1\left(\tau\circ\beta\right)}$ on $J^1\left(\tau\circ\beta\right)$; it results that the contact structure is
\[
  T\left(\tau\circ\beta\right)_{10}-T_xs\circ T\left(\tau\circ\beta\right)_1=\left(de_i^\beta-e^\beta_{i\alpha}dx^\alpha\right)\otimes\frac{\partial}{\partial e^\beta_i}+\left(dv^\beta-v^\beta_{\alpha}dx^\alpha\right)\otimes\frac{\partial}{\partial v^\beta},
\]
and because the infinitesimal generators for the $A\left(m,\mR\right)$-action on $AM$ are
\[
  \left(E_i^j\right)_{AM}\left(x^\alpha,e_i^\beta,v^\gamma\right)=e^\alpha_i\frac{\partial}{\partial e^\alpha_j},\qquad\left(e_i\right)_{AM}\left(x^\alpha,e_i^\beta,v^\gamma\right)=e^\alpha_i\frac{\partial}{\partial v^\alpha},
\]
we obtain that
\begin{equation}\label{eq:LocalCoordinatesCanonicalFormAffine}
  \left.\theta_{J^1\left(\tau\circ\beta\right)}\right|_{\left(x^\alpha,e_i^\beta,v^\rho,e^\beta_{i\gamma},v^\alpha_\beta\right)}=e_\beta^j\left(de_i^\beta-e^\beta_{i\alpha}dx^\alpha\right)\otimes E^i_j+e^i_\beta\left(dv^\beta-v^\beta_{\alpha}dx^\alpha\right)\otimes e_i.
\end{equation}

\printbibliography


\end{document}